\documentclass[12pt]{article}
\usepackage{amsmath}
\usepackage{amsthm}
\usepackage{graphicx}
\usepackage{enumerate}
\usepackage{natbib}
\usepackage[colorlinks=true, linkcolor=blue, urlcolor=blue, citecolor=blue]{hyperref}
\usepackage{multirow} 
\usepackage{dsfont} 
\usepackage{caption}
\usepackage{subcaption}
\usepackage[margin=1in]{geometry}

\theoremstyle{definition}
\newtheorem{theorem}{Theorem}
\newtheorem{lemma}[theorem]{Lemma}

\newtheorem{corollary}{Corollary}

\newtheorem{assumption}{Assumption}

\newtheorem{remark*}{Remark*}

\usepackage{amsfonts} 

\newcommand{\var}{\text{Var}}
\newcommand{\avar}{\text{AVar}}

\newcommand{\redss}{\text{AERSS}}
\newcommand{\are}{\text{ARE}}
\newcommand{\cov}{\text{Cov}}

\newcommand{\unadj}{\text{unadjusted}}

\newcommand{\eff}{\text{efficient}}

\newcommand{\progwl}{{\rm progn}_{W,L}}
\newcommand{\progw}{{\rm progn}_{W}}
\newcommand{\progl}{{\rm progn}_{L}}
\newcommand{\prognon}{{\rm progn}_{\emptyset}}

\newcommand{\itprogwl}{progn_{W,L}}
\newcommand{\itprogw}{progn_{W}}
\newcommand{\itprogl}{progn_{L}}
\newcommand{\itprognon}{progn_{\emptyset}}

\newcommand{\nmax}{n_{\rm max}}

\newcommand{\indic}{{\mathds 1}}

\begin{document}



  \title{Improving Precision through Adjustment for Prognostic Variables in Group Sequential Trial Designs: Impact of Baseline Variables, Short-Term Outcomes, and Treatment Effect Heterogeneity}
  \author{
  Tianchen Qian\thanks{Department of Statistics, Harvard University. qiantianchen@fas.harvard.edu},
  Michael Rosenblum\thanks{Department of Biostatistics, Johns Hopkins University.}, and
  Huitong Qiu\thanks{Vatic Labs, New York.}}
  \maketitle

\bigskip

\begin{abstract}
In randomized trials, appropriately adjusting for baseline variables and short-term outcomes can lead to increased precision and reduced sample size. We examine the impact of such adjustment in group sequential designs, i.e., designs with preplanned interim analyses where enrollment may be stopped early for efficacy or futility. We address the following questions: how much precision gain can be obtained by appropriately adjusting for baseline variables and a short-term outcome?  How is this precision gain impacted by  factors such as the proportion of pipeline participants (those who enrolled but haven't yet had their primary outcomes measured) and treatment effect heterogeneity? What is the resulting impact on power and average sample size in a group sequential design?
We derive an asymptotic formula  that decomposes the overall precision gain from adjusting for baseline variables and a short-term outcome into contributions from factors mentioned above, for efficient estimators in the  model that only assumes randomization and independent censoring.
%
We use our formula to approximate the precision gain from a targeted minimum loss-based estimator applied to data from a completed trial of a new surgical intervention for stroke.
 Our formula implies that (for an efficient estimator) adjusting for a prognostic baseline variable leads to at least as much asymptotic precision gain as adjusting for an equally prognostic short-term outcome.
In many cases, such as our stroke trial application, the former leads to substantially greater precision gains than the latter.
In our simulation study, we show how precision gains from adjustment can be converted into sample size reductions (even when there is no treatment effect). 
\end{abstract}

\noindent%
\vfill

\newpage

\section{Introduction} \label{sec:introduction}

Group sequential designs for randomized clinical trials involve interim analyses that may result in early stopping for efficacy or futility. We consider trial designs where the primary outcome is measured at a fixed time (called the delay) after enrollment.
Prior to observing the primary outcome, participants may have baseline variables and a short-term outcome measured. For example, in the MISTIE-II trial \citep{MISTIE_trial2008} for evaluating a surgical treatment for intracerebral hemorrhage the primary outcome is the modified Rankin Scale (mRS), which measures degree of disability, 180 days after enrollment. A short-term outcome is mRS 30 days after enrollment. The baseline (pre-randomization) variables include age and measures of stroke severity.

Estimators that adjust for baseline variables are used in randomized trials because of the potential for increased precision and reduced sample size.
\citet{pocock2002review}, who surveyed 50 clinical trial reports from major medical journals, found that 36 adjusted for baseline variables. Adjusting for  prognostic baseline variables is recommended by regulators, e.g., \citet{ICH9, EMAguideline2015, FDA-guideline2019}.
Also, in order to address missing outcomes, it  may be useful to appropriately adjust 
 for baseline and post-randomization variables (e.g., short-term outcomes)
 \citep{committee2010national,EMAmissingdata}. 

Various methods for adjustment in randomized trials have been developed; see, e.g., \citet{leon2003semiparametric,davidian2005semiparametric,tsiatis2008covariate, rubin2008covariate, moore2009binary,moore2011relativeefficiency}.
There has also been discussion on the benefits and limitations of adjustment, compared to the standard, unadjusted estimator \citep{yang2001efficiency,freedman2008regression,lin2013agnostic}.
\citet{moore2011relativeefficiency} present a formula for the precision gain from adjusting for baseline variables in single stage trials.



We focus on estimating the average treatment effect for continuous or binary outcomes. We derive a formula for the asymptotic precision gain (measured by the relative efficiency compared to the unadjusted estimator) due to adjustment for baseline variables and a short-term outcome using efficient estimators, that is, estimators that extract the maximum prognostic information from these variables.  
The formula can be used in trial planning to approximate the precision gain from adjustment, which can translate to reduced sample size.

We show, using data of the MISTIE-II trial, how a modified version of these formulas can be used to approximate the precision gain from using the targeted minimum loss-based estimator (TMLE) of \citet{Gruber2012TMLE} while accounting for model misspecification. 
 

We highlight a few implications of our formula. Holding other factors fixed, larger treatment effect heterogeneity decreases the value added from adjusting for baseline variables of all participants; however, for the pipeline participants, adjusting for their baseline variables leads to increased precision only when there is treatment effect heterogeneity. 
 Adjusting for prognostic baseline variables typically leads to precision gains  even if all participants have their primary outcomes  observed. In contrast, adjusting for a prognostic short-term outcome can only improve precision  when there are  participants with the short-term outcome observed but the primary outcome unobserved. 

In Section \ref{sec:notation}, we introduce notation and assumptions. 
In Section \ref{sec:theory-ate}, we present  our formula for precision gain. 
We show how to approximate the precision gain  from using a TMLE in Section \ref{sec:tmle}. 
In Section \ref{sec:application}, we apply our  formula to data from the  MISTIE-II trial.
In Section \ref{sec:simulation}, we conduct simulation studies. Limitations and open problems are discussed in Section \ref{sec:discussion}.

\section{Notation and Assumptions} \label{sec:notation}


For participant $i$, let $A_i$ denote the indicator of study arm assignment. We assume $A_i$ is binary-valued with $A_i=1$ for treatment and $A_i=0$ for control. Denote by $W_i$ a vector of baseline variables measured before randomization. 
All variables must be preplanned in the trial protocol. Denote by $L_i$ a short-term outcome, which is observed at time $d_L$ after enrollment. Though all of our results hold if $L_i$ is any prespecified variable measured after randomization, not necessarily the outcome measured at an earlier time, we focus on the case of short-term outcomes. Also, for simplicity, we focus on cases with binary $L_i$; the results can be generalized to multidimensional, continuous $L_i$. Denote by $Y_i$ the primary outcome (continuous or binary-valued), which is observed at time $d_Y$ after enrollment with $d_Y \geq d_L$. The delays $d_L$ and $d_Y$ are prespecified and are common to all participants. When followed up completely, participant $i$ contributes full data $(W_i, A_i, L_i, Y_i)$. 

Let $n$ denote the sample size. We assume that  the set of vectors  
$\{(W_i, A_i, L_i, Y_i)\}_{1\leq i \leq n}$ are independent and identically distributed draws from the unknown joint distribution $P_0$ on generic data vector $(W, A, L, Y)$. 
The following assumption follows from randomization:
\begin{assumption} \label{assump:randomization}
The study arm assignment $A$ is  independent of the baseline variables $W$ and 
  $P(A=1)=P(A=0)=1/2$.
\end{assumption}
\noindent Results for randomization probabilities other than $1/2$ are in Supplementary Material \ref{proof:thm:avar-ate}.
 
Define the average treatment effect as 
$\Delta = E(Y|A=1)-E(Y|A=0)$. 
The goal is to test the null hypothesis of no average treatment benefit: $H_0: \Delta \leq 0$ versus $H_1: \Delta > 0$.

We assume that participants are enrolled at a constant rate. Since not all participants have full data observed at an interim analysis, we introduce indicators $C^L$ and $C^Y$ to denote  that $L$ and $Y$, respectively, have been observed at or before a given analysis time. For a participant, $C^L = 1$ if and only if $L$ is observed, and $C^Y = 1$ if and only if $Y$ is observed.
 These variables get updated at each analysis. 
We make the following assumptions:
\begin{assumption} \label{assump:independent-censoring}
(Independent Censoring) $(C^L, C^Y)$ is independent of $(W,A,L,Y)$.
\end{assumption}
\begin{assumption} \label{assump:monotone-censoring}
(Monotone Censoring) $C^L \geq C^Y$ with probability $1$.
\end{assumption}
\noindent An example where Assumptions \ref{assump:independent-censoring} and \ref{assump:monotone-censoring} hold is when administrative censoring is the only source of missingness and there are no changes over time in the population who are enrolled; this is what we simulate in Section~\ref{sec:simulation}. 
At any analysis time, each enrolled participant has one of the following missingness patterns: 
\begin{itemize}
    \item[(i)] $(C^L, C^Y) = (0,0)$: a pipeline participant with only $(W,A)$ observed;
    \item[(ii)] $(C^L, C^Y) = (1,0)$: a pipeline participant with $(W,A,L)$ observed;
    \item[(iii)] $(C^L, C^Y) = (1,1)$: a participant with $(W,A,L,Y)$ observed.
\end{itemize}

We assume a nonparametric model for the joint distribution of the variables $(W, A, L, Y,C_L,C_Y)$, except that we make Assumptions 1-3.
The semiparametric efficiency results in the following sections are with respect to this model.  

For a group sequential design with $K$ stages, we consider the asymptotic setting where the maximum sample size goes to infinity such that at each interim analysis the proportions of enrolled participants with final and short-term outcomes observed converge to constants $p_y$ and $p_l$, respectively (where these limit proportions may differ by analysis time, but for conciseness we suppress their dependence on the analysis time). This asymptotic regime corresponds to fixing the delay times $d_L, d_Y$ and analysis times, while increasing the enrollment rate. We only consider analysis times with $p_y > 0$. 

The unadjusted estimator of $\Delta$ is the difference between the two arms of the sample mean of the primary outcome $Y$, using data from all participants with $Y$ observed, i.e., 
$\sum_i A_i C^Y_i Y_i/\sum_i A_i C^Y_i - \sum_i (1-A_i) C^Y_i Y_i/\sum_i (1-A_i) C^Y_i.$

All estimators of the average treatment effect $\Delta$ that we consider are regular, asymptotically linear (RAL), defined, e.g., by \cite{BKRW1993}. Any such estimator $\widehat{\Delta}$ of $\Delta$ is asymptotically normal, i.e.,  $\sqrt{n}(\widehat{\Delta} - \Delta) \stackrel{d}{\to} N(0, \sigma^2)$ as sample size $n$ goes to infinity, where $\stackrel{d}{\to}$ denotes convergence in distribution; the variance $\sigma^2$ is called the 
 asymptotic variance of $\widehat{\Delta}$ and is denoted by $\avar(\widehat{\Delta})$.
 For example, the asymptotic variance of the unadjusted estimator is  $2\{\var(Y|A=1)+\var(Y|A=0)\}/p_y$
  at an analysis time where $p_y$ fraction of the enrolled participants have their primary outcome observed.
The asymptotic relative efficiency (ARE) between two RAL estimators $\widehat{\Delta}_1$ and $\widehat{\Delta}_2$ of $\Delta$ is the inverse of the ratio of their asymptotic variances:
$\are(\widehat{\Delta}_1, \widehat{\Delta}_2) = \avar(\widehat{\Delta}_2)/\avar(\widehat{\Delta}_1).$ 

For any random vector $X$ and $a\in\{0,1\}$, let $E_a(Y \mid X)$ and $\var_a(Y \mid X)$ denote $E(Y\mid X, A=a)$ and $\var(Y\mid X, A=a)$, respectively. For example, $E_1(Y \mid W) = E(Y \mid W, A=1)$, and $E_1(Y) = E(Y \mid A=1)$. When no subscript is used on $E$ or $\var$, these refer to expectation and variance, respectively, under the population distribution of the corresponding variables.

\section{Best Possible Precision Gain}
\label{sec:theory-ate}

\subsection{Formula for Precision Gain}

The following result gives the best possible asymptotic variance that can be achieved by a RAL estimator of the average treatment effect $\Delta$, in terms of the proportions $p_l, p_y$ of enrolled participants with $L$ and $Y$ observed, respectively, at a given analysis time:

\begin{lemma} 
\label{thm:avar-ate}
Suppose Assumptions \ref{assump:randomization}, \ref{assump:independent-censoring}, and \ref{assump:monotone-censoring} hold. The asymptotic variance of any RAL estimator of $\Delta$ is at least 
\begin{align}
   & \var\{E_1(Y \mid W) - E_0(Y \mid W) \}
+  \sum_{a\in\{0,1\}} \frac{2}{p_l} \var_a\{ E_a(Y \mid L,W) - E_a(Y \mid W) \} \nonumber \\
+ &  \sum_{a\in\{0,1\}} \frac{2}{p_y} \var_a\{ Y - E_a(Y \mid L,W) \}. \label{eq:thm:var-bound-ate}
\end{align}

\end{lemma}

Lemma \ref{thm:avar-ate} is a consequence of the efficient influence function of $\Delta$ in our  semiparametric model that only makes the assumptions  in Section~\ref{sec:notation}. 
This efficient influence function is given in 
 Section~\ref{proof:thm:avar-ate} of the Supplementary Material, and follows from  \citet{robins1992semipara,Scharfstein1999,van2003semipara}.
The TMLE estimator of \citet{Gruber2012TMLE}, which will be discussed in Section \ref{sec:tmle}, achieves this variance lower bound when all working models are correct. When certain working models for this estimator are misspecified, one can approximate the corresponding $\are$ by substituting regression model fits for conditional expectations in (\ref{eq:thm:var-bound-ate}); see Section~\ref{subsec:are-fitted}.


The first term in (\ref{eq:thm:var-bound-ate}) characterizes the variance in the conditional treatment effect across different levels of $W$. We define the \textit{treatment effect heterogeneity}, denoted by $\gamma$, by dividing the variance of the conditional treatment effect by the sum of variances of $Y$ in each arm:
\begin{align}
\gamma = \frac{\var\{E_1(Y \mid W) - E_0(Y \mid W) \}}{\sum_{a\in\{0,1\}} \var_a(Y)}. \label{def:gamma}
\end{align}
The treatment effect heterogeneity $\gamma$  is invariant to linear transformations of $Y$, and is nonnegative. 
When $\gamma = 0$ there is no treatment effect heterogeneity, i.e., the conditional treatment effect $E_1(Y \mid W)-E_0(Y \mid W)=\Delta$ with probability 1.


For each arm $a \in \{0,1\}$, the variance of $Y$ given $A=a$ can be decomposed as follows:
\begin{equation}
\var_a(Y) = \var_a\{Y - E_a(Y\mid L,W)\} +  \var_a\{ E_a(Y\mid L,W) - E_a(Y \mid W) \} + \var\{E_a(Y\mid W)\}, \label{var_decomp_simple}
\end{equation}
as proved in Section~\ref{proof:lem:vardecomp} in the Supplementary Material. The last term in the display above motivates 
 the following definition of the proportion of the variance in $Y$ explained by $W$ (summed across arms): 
\begin{equation}
R^2_W = \frac{\sum_{a\in\{0,1\}} \var\{E_a(Y \mid W)\} }{\sum_{a\in\{0,1\}} \var_a(Y)}  = 1 - \frac{\sum_{a\in\{0,1\}} \var_a\{Y - E_a(Y \mid W)\} }{\sum_{a\in\{0,1\}} \var_a(Y)} 
.
\end{equation}
Similarly, the middle term in (\ref{var_decomp_simple}) motivates the following definition of the proportion of additional variance in $Y$ explained by $L$ after accounting for $W$ (summed across arms):
\begin{equation}
R^2_{L \mid W} = \frac{\sum_{a\in\{0,1\}} \var_a\{E_a(Y \mid L, W) - E_a(Y \mid W)\} }{\sum_{a\in\{0,1\}} \var_a(Y) }. \label{def:RsquaredLgivenW}
\end{equation}

A RAL estimator of $\Delta$ is efficient (at a given data generating distribution) if it achieves the asymptotic variance lower bound (\ref{eq:thm:var-bound-ate}) in the semiparametric model. 

\begin{theorem} 
\label{cor:re-ate}

Suppose Assumptions \ref{assump:randomization}, \ref{assump:independent-censoring}, and \ref{assump:monotone-censoring} hold.  Consider any analysis time.  The ARE between any  efficient RAL estimator of $\Delta$ and the unadjusted estimator is 
\begin{equation}
\are = \frac{1}{ 1 + (p_y/2) \gamma - R^2_W -  (1 - p_y / p_l) R^2_{L\mid W} }. \label{eq:re-ate}
\end{equation}
\end{theorem}
\noindent The denominator of the right side of the above display never exceeds 1, which follows since $\gamma \leq 2 R^2_W$ by the Cauchy-Schwarz inequality.


\subsection{Impact of Baseline Variables on Precision Gain} \label{sec:impactbaseline}


In order to isolate the impact of baseline variables $W$ on the ARE \eqref{eq:re-ate}, we consider the case where there is no impact of $L$, i.e., when 
$R^2_{L \mid W} = 0$. This is the case, for example, when $L$ is independent of $Y$ given $A$ and $W$.
Then the $\are$ from adjusting for baseline variables $W$ depends on three factors: the proportion of variance in $Y$ explained by $W$ ($R^2_W$), the proportion of participants with $Y$ observed among those enrolled ($p_y$), and the treatment effect heterogeneity ($\gamma$). 
We plot in Figure \ref{fig:are-w-ate} the $\are$ versus $p_y$, considering different combinations of $R^2_W$ and $\gamma$.

\begin{figure}[htbp]
    \caption{Asymptotic relative efficiency between an efficient estimator and the unadjusted estimator for estimating $\Delta$, when only the baseline variable $W$ is prognostic ($R^2_{L \mid W} = 0$ in all curves).}
    \label{fig:are-w-ate}
\begin{center}
    \includegraphics[width=0.8\textwidth]{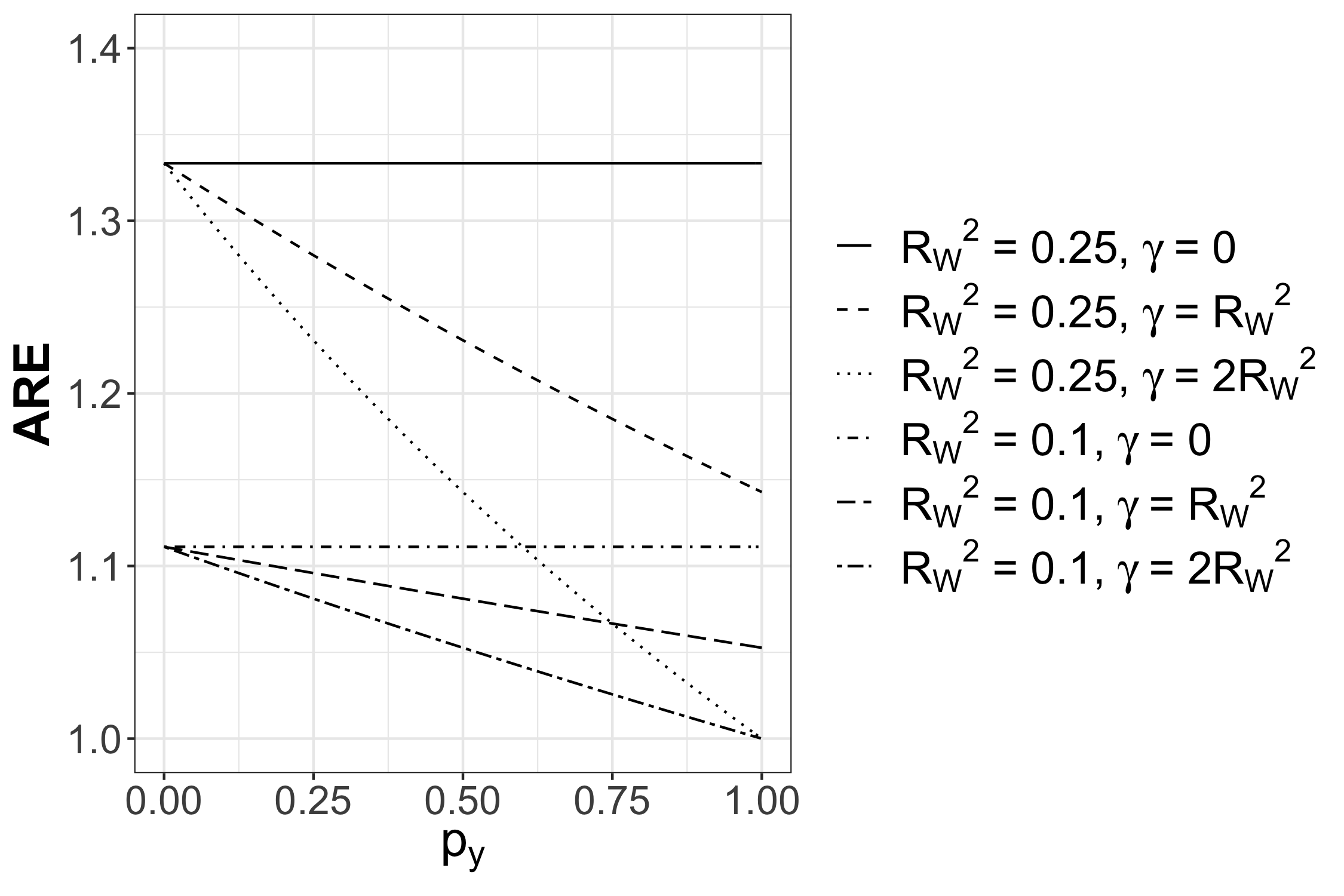}
\end{center}
\end{figure}

 A summary of what is happening in Figure~\ref{fig:are-w-ate} is that with $p_y=1$ and $R^2_W$ fixed, greater treatment effect heterogeneity ($\gamma$) lessens the precision gain from prognostic baseline variables. 
 However, with fixed $R^2_W>0$ and $\gamma>0$, some precision gain is restored as we decrease $p_y$ from $1$ to $0$ (i.e., as we move from right to left along any line with $\gamma>0$ in Figure~\ref{fig:are-w-ate}). 
   We describe the intuition for these phenomena below. The key idea is that precision gains result from adjusting for chance imbalances in $W$ both (i) between study arms among the participants with $Y$ observed, and (ii) between the full sample and the subset with $Y$ observed. 
 Whether the adjustments in (i) and (ii) lead to precision gains or not is determined by the treatment effect heterogeneity $\gamma$, which influences whether  the estimated means in each arm are adjusted in opposite directions (leading to variance reduction) or in the same direction (leading to cancellation and no impact). This is analogous to constructive versus destructive interference between waves, which can lead to cancellation or amplification as the waves come together.



To give intuition for the phenomena in Figure~\ref{fig:are-w-ate}, for the remainder of this subsection we fix $R^2_W$ and consider the simple case of a single, binary-valued $W$ representing being less than 65 years old at baseline, and primary outcome $Y$  being the indicator of having mRS at most 3 (a good outcome) at 180 days.
 We focus on the stratum $W=1$ and assume that 
the under 65 subset of the study population would have better outcomes on average than the overall study population if all were assigned to the control arm, i.e.,  $E(Y|A=0,W=1)-E(Y|A=0)>0$. 
We next consider $p_y=1$ and the opposite extremes of no treatment effect heterogeneity (where efficiency gains from adjusting for $W$ are maximal) and maximum treatment effect heterogeneity (where there are no efficiency gains from adjusting for $W$). 


First, consider the case of $p_y=1$ and no treatment effect heterogeneity ($\gamma=0$), i.e., 
 $E(Y|A=1,W)-E(Y|A=0,W)$ equals the constant $E(Y|A=1)-E(Y|A=0)=\Delta$. 
 Combining this with the assumption above that those under 65 have better outcomes on average than the overall study population under assignment to the control arm ($E(Y|A=0,W=1)-E(Y|A=0)>0$), it follows that those under 65 also have better outcomes on average than the overall study population under assignment to the {\em treatment} arm, i.e.,  
$E(Y|A=1,W=1)-E(Y|A=1)=E(Y|A=0,W=1)-E(Y|A=0)$. 
 If by chance there are proportionally more under 65 participants assigned to the treatment arm than the control arm, then the unadjusted estimator of the mean outcome in the treatment (control) arm is biased upward (downward) conditional on the chance imbalance \citep{jiang2019robust}. An efficient estimator adjusts for this by decreasing the  
 unadjusted estimate of $E(Y|A=1)$ and by increasing the unadjusted estimate of $E(Y|A=0)$; the net impact is to decrease the estimate of $\Delta$.  Alternatively, if by chance there are proportionally fewer under 65 participants assigned to the treatment arm, then by a symmetric argument an efficient estimator adjusts for this by increasing the 
 unadjusted estimate of $\Delta$. The overall impact of such adjustments across many hypothetical trials, is to remove the variance caused by chance imbalances across arms in the stratum $W=1$. This is, intuitively, why there is an efficiency gain when $\gamma=0$.
 
 Now consider $p_y=1$ and the opposite extreme of maximum  treatment effect heterogeneity ($\gamma=2R^2_W$), which occurs when $E(Y|A=1,W)+E(Y|A=0,W)$ equals the constant $E(Y|A=1)+E(Y|A=0)$. An analogous argument (given in Section~\ref{appen:adjustmentintuition} of the Supplementary Material) to the previous paragraph shows that adjustment for chance imbalances in $W$ leads to equal changes (in the same direction) in the estimated means in each arm; the net impact is that the difference between estimated means across arms is not changed (due to cancellation). Adjustment has no impact in this case, and leads to no variance reduction.

Above, we considered $p_y=1$, i.e., no pipeline participants. We now fix $R^2_W>0$ and $\gamma>0$, and consider the impact of decreasing $p_y$ from $1$ to $0$. 
This leads to increased $\are$ because an efficient estimator can extract information from pipeline participants if  $\gamma>0$; in contrast, the unadjusted estimator does not use any information from pipeline participants. Intuition for why $\gamma>0$ is necessary in order for pipeline participant information to be useful 
is 
given  in Section~\ref{appen:adjustmentintuition} of the Supplementary Material.



\subsection{Impact of a Short-term Outcome on Precision Gain} \label{impactshortterm}
To isolate the impact of a short-term outcome $L$ on precision gain, we  set $R^2_W = 0$ in \eqref{eq:re-ate} so that only $L$ is prognostic.
The precision gain from adjusting for short-term outcome $L$ depends on two factors: the proportion of variance in $Y$ explained by $L$ after accounting for $W$ ($R^2_{L \mid W}$), and the proportion of participants with $Y$ observed among those with $L$ observed ($p_y / p_l$). 
In Figure \ref{fig:are-l-ate}, we plot the $\are$ against $p_y / p_l$, and 
 consider two values of $R^2_{L \mid W}$. 

\begin{figure}[htbp]
 \caption{Asymptotic relative efficiency between an efficient estimator and the unadjusted estimator for estimating $\Delta$, when only the short-term outcome $L$ is prognostic ($R^2_W = 0$ in both curves). The solid line corresponds to a higher prognostic value $R^2_{L \mid W}=0.25$ and the dashed line corresponds to a lower prognostic value $R^2_{L \mid W}=0.1$.
    } 
    \label{fig:are-l-ate}
\begin{center}
 \includegraphics[width=0.48\textwidth]{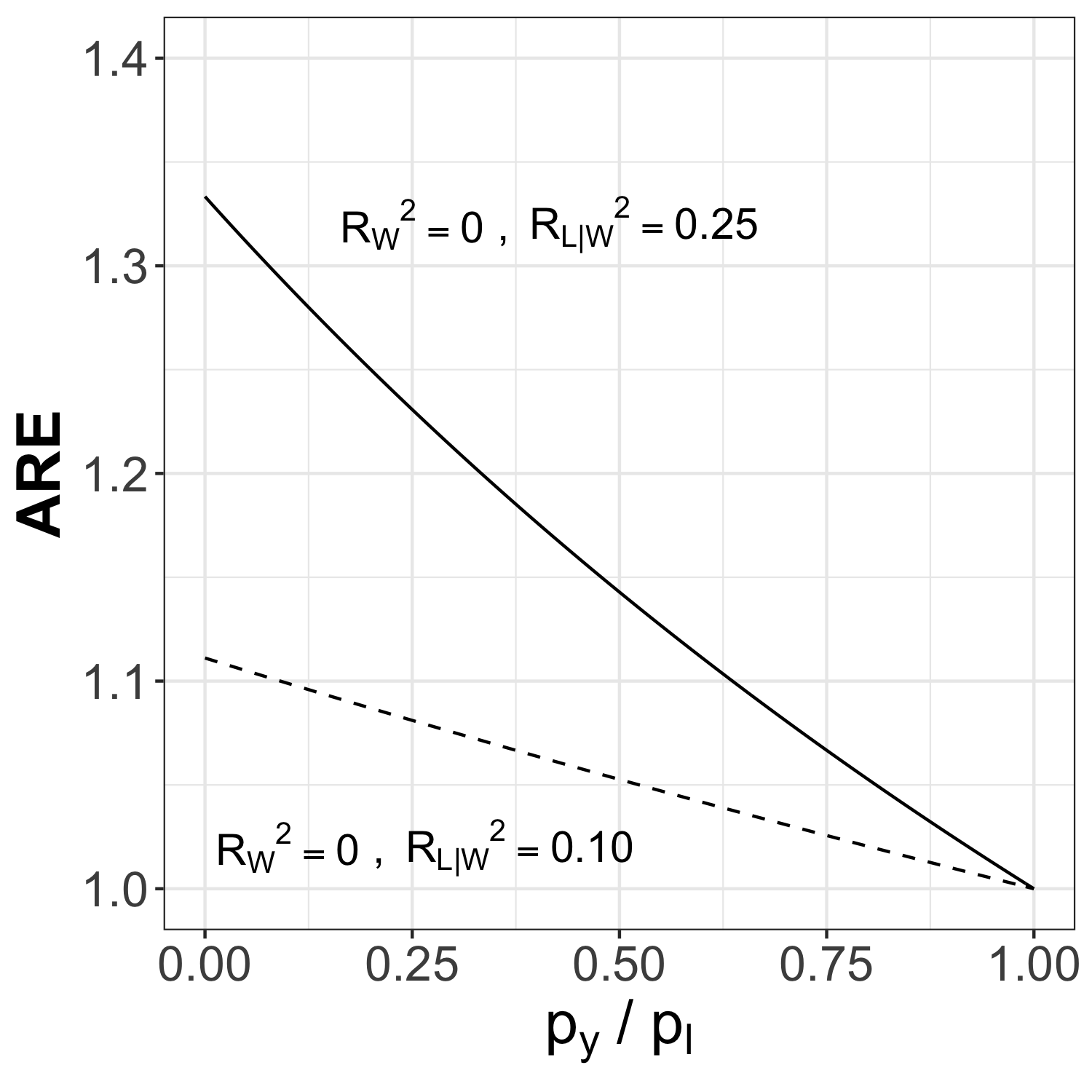}
\end{center}
\end{figure}

Smaller $p_y/p_l$ and larger $R^2_{L \mid W}$ generally increase the precision gain from adjusting for $L$; this is because $L$ adds value just for participants with $L$ but not $Y$ observed, 
and $R^2_{L \mid W}$ quantifies the prognostic value of $L$ beyond the variance in $Y$ explained by $W$. 
When there are no pipeline participants, which would occur at the end of a trial with no early stopping or dropout, 
 we have $p_y/p_l = 1$ and 
the $\are$ is 1, i.e., adjusting for $L$ is useless.
This is because adjusting for $L$ helps by, in each arm separately, accounting for chance imbalances between the  participants with $Y$ observed and the participants with $L$ but not $Y$ observed; when  $p_y/p_l = 1$, these groups are identical and no adjustment can be made.

Recall that we are operating in the semiparametric model defined by Assumptions 1-3; in particular, we are not making any assumptions about the relationships among the variables $(Y,A,L)$. Therefore, we cannot adjust for observed imbalances in $L$ between arms $A \in \{0,1\}$, since these imbalances may be due to the impact of arm assignment. 

Given fixed $R^2_{L \mid W}>0$, larger $p_y / p_l$ decreases the precision gain from adjusting for $L$, as seen in Figure~\ref{fig:are-l-ate}. In a trial with constant enrollment rate and $d_L < d_Y$, the precision gain from adjusting for $L$ typically attenuates at later stages of a group sequential design due to $p_y / p_l$ being  nondecreasing over time. In the case where missingness is only due to administrative censoring, $p_y / p_l$ starts at $0$ at time $d_L$ and eventually becomes 1 after enrollment stops and enough time has elapsed for all enrolled participants to have $Y$ observed.
 
\subsection{Comparison of Equally Prognostic Baseline Variable and Short-term Outcome} \label{impactcomparison}
For any $q: 0 < q \leq 1$, we compare the $\are$ between two cases: $R^2_W = q, R^2_{L \mid W} = 0$ (only baseline variable prognostic) and 
 $R^2_W = 0, R^2_{L \mid W} = q$ (only short-term outcome prognostic).  
 The $\are$ in the 
 former case is larger or equal to that in the latter case. Equality occurs 
  if and only if  
 $p_l=1$ and treatment effect heterogeneity in the first case is the maximum possible ($\gamma=2R^2_W=2q)$, as
 proved 
in Section~\ref{appen:proof-section3.4} of the Supplementary Material. Intuition for this is given in  Section~\ref{appen:adjustmentintuition} of the Supplementary Material.
  The equality can be seen graphically in that 
the two lines in Figure~\ref{fig:are-l-ate} (corresponding to prognostic value $0.25$ and $0.1$) are identical to the two 
  corresponding lines in Figure~\ref{fig:are-w-ate} with $\gamma=2R^2_W$; the horizontal axes of the two figures are identical in the special case of $p_l=1$.



\section{Estimator that Adjusts for Baseline Variables and Short-term Outcomes} \label{sec:tmle}


\subsection{Targeted Minimum Loss-Based Estimator}
 The targeted minimum loss-based estimator (TMLE) of  \cite{Gruber2012TMLE}, which 
 builds on the ideas of \citet{robins1992semipara,Scharfstein1999, 
 Robins2000TMLE,BangRobins2005TMLE,van2006tmle}, 
 is implemented in the \textsf{R} package \textsf{ltmle} \citep{ltmlePackage}. 
   We use it to estimate $\Delta$ and call it the \textit{adjusted estimator}.

Our implementation of this TMLE for a binary outcome $Y$ involves fitting logistic regression  working models.  Such models are fit first for the conditional probabilities of censoring given the observed history before censoring: $P(C^L=1 \mid A, W)$ and $P(C^Y=1 \mid L, A, W)$.  We also fit such a model for study arm assignment $P(A=1 \mid W)$. Lastly, we fit logistic regression models for $E(Y| L,A,W)$ and for $E(Y| A,W)$, given the corresponding observed histories; we refer to these as outcome regression models. The last model is fit  using sequential regression, an idea from  \cite{Robins2000TMLE,BangRobins2005TMLE}; see our R code (Github link given at end of paper) for the implementation of the software of \cite{ltmlePackage} that we used. All models involve an intercept and main terms for each regressor variable. Each regression is performed using all participants for whom the relevant variables are uncensored.
For a continuous-valued outcome $Y$, linear outcome regression models could be used or, if the outcome is bounded then it can be rescaled to the interval $[0,1]$ and logistic regression can be used as described by  \cite{grubervanderLaan2010}.

The adjusted estimator initially uses the aforementioned regression model fits. It then updates each outcome regression model by adding a new covariate built from the other model fits, as described by  \cite{Gruber2012TMLE}. The final estimator of the average treatment effect $\Delta$ is based on the updated outcome regression model fit of $E(Y|A,W)$. It is computed by first generating a prediction $Y_{a,i}$ of the outcome $Y$ under each hypothetical assignment to arm $a \in \{0,1\}$ for each participant $i$ by substituting $W=W_i$ and $A=a$ into this model fit. The estimator of $\Delta$ is $(1/n)\sum_{i=1}^n (Y_{1,i}-Y_{0,i})$. 


Assumptions 1-3 from Section~\ref{sec:notation} imply that the censoring models and the model for $P(A=1|W)$ are correctly specified. 
The results of \citet[Section 4]{Gruber2012TMLE} imply that the adjusted estimator is a consistent estimator of the average treatment effect $\Delta$.
This holds regardless of the correctness of the outcome regression models. That is, these models could be arbitrarily misspecified and still the adjusted estimator converges in probability to $\Delta$ as sample size goes to infinity.
 If the outcome regression models are correct, then the adjusted estimator achieves the asymptotic variance lower bound in (\ref{eq:thm:var-bound-ate}), in which case it is semiparametric, locally efficient \citep[Section 4]{Gruber2012TMLE}. 
 Other estimators with the above properties include the augmented, inverse probability of treatment (and censoring) estimators of \citet{Robins2000TMLE} and \citet{BangRobins2005TMLE}.


\subsection{Approximating the Relative Efficiency between the Adjusted Estimator and the Unadjusted Estimator} \label{subsec:are-fitted}

We focused on the asymptotic relative efficiency between an efficient estimator and the unadjusted estimator in Section \ref{sec:theory-ate}. 
To connect our results regarding efficient estimators to the adjusted estimator, we use results from \citet[Section 5.3]{Gruber2012TMLE}. These results imply that if all regression models used in the TMLE are correctly specified, then 
 plugging the corresponding model fits into the formulas (\ref{def:gamma})-(\ref{def:RsquaredLgivenW}) and then using these estimates of $\gamma,R^2_W,R^2_{L \mid W}$ in (\ref{eq:re-ate}) results in a consistent estimator of the $\are$ (called the plug-in estimator).
 This estimator of the $\are$ is asymptotically conservative, meaning that the $\are$s may be underestimated but not overestimated, if 
  the models for censoring and for $P(A=1|W)$ are correctly specified; since these models are correct under Assumptions \ref{assump:randomization}-\ref{assump:monotone-censoring}, 
it follows that the plug-in estimator can serve as a (possibly conservative) asymptotic approximation to the $\are$ between the adjusted estimator and the unadjusted estimator.

Denote by $\widehat{E}(Y \mid X=x)$ the predicted value of  $E(Y \mid X=x)$ from a model fit. 
Let $\widehat{\var}(\cdot)$ denote the sample variance over participants from both arms, and $\widehat{\var}_a(\cdot)$ denote the sample variance over participants from arm $a$.
We estimate $R^2_W$, $R^2_{L\mid W}$, and  $\gamma$, respectively, by
\begin{equation*}
	\widehat{R}^2_W  = \frac{ \sum_{a\in\{0,1\}} \widehat{\var}\{ \widehat{E}(Y \mid W, A=a) \}}{ \sum_{a\in\{0,1\}} \widehat{\var}_a(Y) },
\end{equation*}
\begin{equation*}
	\widehat{R}^2_{L \mid W}  = \frac{ \sum_{a\in\{0,1\}} \widehat{\var}_a\{\widehat{E}(Y \mid W, L, A=a) - \widehat{E}(Y \mid W, A=a)\} }{ \sum_{a\in\{0,1\}} \widehat{\var}_a(Y) },
\end{equation*}
\begin{equation*}
    \widehat\gamma = \frac{\widehat\var\{ \widehat{E}(Y \mid W, A=1) - \widehat{E}(Y \mid W, A=0) \}}{\sum_{a\in\{0,1\}} \widehat{\var}_a(Y)}.
\end{equation*}
In Sections~\ref{sec:application}-\ref{sec:simulation} 
when calculating the conditional expectations $\widehat{E}$ in the formulas for $\widehat{R}^2_W$ and $\widehat{R}^2_{L \mid W}$, we used logistic regression model fits with main terms.  When calculating $\widehat\gamma$, we used  similar model fits except including interactions, since we wanted to give an opportunity to detect treatment effect heterogeneity.  
The estimated $R$-squared quantities above are relevant for continuous and binary outcomes in measuring the precision gain from adjustment when estimating $\Delta$; this generalizes an idea from \citet{moore2009binary}, 
who give an R-squared formula related to $\widehat{R}^2_W$ except that it marginalizes over $A$.



\section{Applying Precision Gain Formula to Trial Data}
\label{sec:application}


\subsection{Trial Example: MISTIE-II} \label{subsec:mistie}

MISTIE-II is a Phase II randomized clinical trial evaluating a new surgical treatment for intracerebral hemorrhage. The treatment is called Minimally-Invasive Surgery Plus rt-PA for Intracerebral Hemorrhage Evacuation, abbreviated as MISTIE \citep{MISTIE_trial2008}. In the MISTIE-II dataset, the primary and short-term outcomes of each participant correspond to the Modified Rankin Scale (mRS), measured at different times after enrollment. We use the variables in Table \ref{tab:MISTIEvars}. The primary outcome $Y$ is a binary indicator of a successful outcome ($\mbox{mRS}\leq3$) at 180 days after enrollment. 
The short-term outcomes $L^{(1)}$ and $L^{(2)}$ are the indicators of mRS no larger than 4 at 30 and 90 days after enrollment, respectively. The cutoff 4 for $L^{(1)}$ and $L^{(2)}$ was chosen because there were relatively few participants with mRS 3 or less at 30 or 90 days after enrollment.
The treatment assignment indicator $A$ denotes whether a participant is assigned to the surgical treatment ($A=1$) or to standard of care ($A=0$). Baseline variables $W^{(1)}$-$W^{(4)}$ are age (dichotomized at 65 years), NIHSS (NIH Stroke Scale total score, quantifying stroke-caused impairment), ICH (intracerebral hemorrhage volume), and GCS (Glasgow Coma Scale), all measured before randomization. The baseline variables (except for age) are treated as continuous variables in the regression models. The dataset has 100 participants.


\begin{table}[htbp] 
\caption{Variables that we use from the MISTIE-II dataset. Ordinal-valued basline variables below are treated as continuous in our regression models. 
}
\label{tab:MISTIEvars}
\begin{center}
\begin{tabular}{lll}
\hline
Name 			& Description              				\\
\hline
$W^{(1)}$		& baseline age (years), dichotomized at  $\leq 65$    		\\
$W^{(2)}$		& baseline NIHSS, ordinal           \\
$W^{(3)}$		& baseline ICH, ordinal 			    \\
$W^{(4)}$		& baseline GCS, ordinal   		    \\
$A$    			& treatment indicator, 1 being MISTIE               	\\
$L^{(1)}$		& mRS at 30 days, dichotomized by thresholding at  $\leq 4$     \\
$L^{(2)}$    	& mRS at 90 days, dichotomized by thresholding at  $\leq 4$ 	\\
$Y$        		& mRS at 180 days, dichotomized by thresholding at  $\leq 3$   \\ \hline
\end{tabular}
\end{center}
\end{table}

We let $W^\text{full} = (W^{(1)}, W^{(2)}, W^{(3)}, W^{(4)})$ and $L^\text{full} = (L^{(1)}, L^{(2)})$. These variables are used in constructing data generating distributions for our simulations  in Section \ref{sec:simulation}. However, only the smaller subsets of variables $W = (W^{(1)}, W^{(4)})$ and $L = L^{(1)}$ are made available to the adjusted estimator (both in  Section~\ref{subsec:application-precision-gain} and~\ref{sec:simulation}). 
We included the ``extra'' variables $W^{(2)}, W^{(3)}, L^{(2)}$ (which the adjusted estimator does not have access to) in our data generating distributions for simulations 
so that the regression models used by the adjusted estimator will be misspecified; this was done since we expect at least some misspecification to occur in practice.

\subsection{Approximate Precision Gain by Using the Adjusted Estimator on MISTIE-II Data}
\label{subsec:application-precision-gain}


Applying the method in Section \ref{subsec:are-fitted} to the MISTIE-II data, the estimated R-squared quantities are $\widehat{R}^2_W = 0.36$, $\widehat{R}^2_{L \mid W} = 0.08$, and $\widehat\gamma = 0.02$. This indicates that the baseline variables are strongly prognostic, that (after accounting for the baseline variables) the short-term outcome $L$ is only mildly prognostic, and that there is little (if any) treatment effect heterogeneity. 

At the end of a trial, if every participant has $Y$ observed, we have $p_l = p_y = 1$. Substituting these and the estimated values $\widehat{R}^2_W, \widehat{R}^2_{L \mid W}, \widehat{\gamma}$ into (\ref{eq:re-ate}) gives an approximate $\are$ of 
$1.53$ comparing the TMLE versus the unadjusted estimator. If instead there had been no treatment effect heterogeneity ($\widehat \gamma = 0$), then the approximate $\are$ would be slightly larger  ($1.55$), assuming $\widehat R^2_W$ and $\widehat R^2_{L \mid W}$ remain fixed. Now consider a hypothetical  interim analysis where 80\% of the enrolled participants have $Y$ observed and 95\% of the enrolled participants have $L$ observed (e.g., at the 3rd interim analysis of the group sequential design with $\nmax = 480$ considered in Section \ref{subsec:simulationresults-2}; see Table \ref{tab:sample-size} in the Supplementary Material). Then we have $p_y = 0.8$, $p_l = 0.95$, and the approximate $\are$ is $1.56$. 

Equivalently, in each of the three aforementioned situations, the adjusted estimator approximately reduces the required sample size to achieve a desired power by $35\%$, $36\%$, and $36\%$, respectively, compared to the unadjusted estimator. (The asymptotic sample size reduction is computed as $1-1/\are$ as described in Section~\ref{sec:theory-singlearm} of the Supplementary Material.) These are all quite similar, which is primarily due to the relatively small values of $\widehat \gamma$ and $\widehat R^2_{L \mid W}$. 
The magnitudes of the AREs above are not unusual in the MISTIE-II population--see \citet{optimising2009should}--and result from some baseline variables (especially NIHSS) being strongly prognostic for the outcome.


Approximations of the $\are$ as above can be applied to trial planning, where one may use previous data to estimate the variance explained by baseline variables and short-term outcomes and the treatment effect heterogeneity, and use \eqref{eq:re-ate} to get a rough projection of the precision gain from adjusting for baseline variables and short-term outcomes. If treatment effect heterogeneity is difficult to estimate \textit{a priori}, one could use $\gamma = 0$ as the best case scenario (no treatment effect heterogeneity, corresponding to maximal precision gain), with the expectation that $\gamma > 0$ will result in attenuated precision gain. 

A conservative approach is to assume no precision gain when selecting the initial sample size per stage, but then use preplanned sample-size re-estimation (using the estimated variance) to potentially shrink the per-stage sample size  to reflect the estimated precision gains based on accrued data. 
Such adaptations based on estimating a nuisance parameter are generally acceptable to regulators such as the U.S. Food and Drug Adminsitration  \citep{fda:2016,FDA}. This approach has important  limitations given in Section~\ref{sec:discussion}.


\section{Simulations of Group Sequential Design} \label{sec:simulation}


\subsection{Data Generating Distributions for Simulated Trials Based on MISTIE-II}
\label{subsec:dgm}

We conducted simulation studies to assess how good of an approximation is provided by the  precision gain formula (\ref{eq:re-ate}), and to examine the sample size reduction by using the adjusted versus unadjusted estimator in a group sequential design. In order to mimic key features of the MISTIE-II trial, we use a resampling-based algorithm to generate participants for our simulated trials. This algorithm generates simulated trials with the following properties:
\begin{enumerate}
	\item[(i)] The treatment assignment is independent of baseline variables.
	\item[(ii)] The relative efficiency between the adjusted estimator and the unadjusted estimator is similar to that calculated directly from the MISTIE-II dataset.
\end{enumerate}


We define eight data generating distributions, called ``settings''. The goal is to consider four situations involving $W$ and/or $L$ being prognostic or not for $Y$; these are denoted by $\itprogwl$ (both prognostic),  $\itprogw$ (only $W$ prognostic), $\itprogl$ (only $L$ prognostic),  and  $\itprognon$ (neither prognostic).  For each of these four situations (called ``prognostic settings''), we construct distributions with the following two average treatment effects:  $\Delta=0$ (no effect) or $\Delta=0.122$ (benefit equal to the unadjusted estimate from the MISTIE-II dataset).
Details of the data generating algorithm for the eight settings are in Section \ref{appen:dgm} of the Supplementary Material. 
Table \ref{tab:rsq-mistie} gives the approximate $R$-squared and $\gamma$ values for the eight settings. The only source of missing data in our simulations is  administrative censoring, i.e., participants who enrolled but have not been in the trial long enough to have outcomes measured.

\begin{table}[htbp]
\caption{Approximate $R$-squared and $\gamma$ for each of the eight data generating distributions (settings) in our simulation study, computed by simulating $10^6$ participants under each distribution and computing $\widehat{R}^2_W, \widehat{R}^2_{L \mid W}, \widehat{\gamma}$, respectively.}
\label{tab:rsq-mistie}
\begin{center}
\begin{tabular}{llllllllllll}
\hline
                 & \multicolumn{2}{l}{$\progwl$} & & \multicolumn{2}{l}{$\progw$} & & \multicolumn{2}{l}{$\progl$} & & \multicolumn{2}{l}{$\prognon$} \\
                 & $\Delta=0$         & $\Delta=0.122$       &  & $\Delta=0$         & $\Delta=0.122$       & & $\Delta=0$         & $\Delta=0.122$        & & $\Delta=0$          & $\Delta=0.122$         \\
\hline
$R^2_{W}$        & 0.35          & 0.36        &  & 0.35          & 0.36        & & 0             & 0            & & 0              & 0             \\
$R^2_{L \mid W}$ & 0.08          & 0.07         & & 0             & 0          &  & 0.30          & 0.30         & & 0              & 0             \\
$\gamma$         & 0             & 0.01         & & 0             & 0.01        & & 0             & 0            & & 0              & 0     \\
\hline
\end{tabular}
\end{center}
\end{table}

%

\subsection{Group Sequential Trial Design Analysis Timing and Early Stopping Rule} \label{subsec:design}


We assume that $d_L = 30$ days and $d_Y = 180$ days. 
The goal is to control Type I error at level 0.025 and to have at least 80\% power to reject the null hypothesis $H_0: \Delta \leq 0$ when $\Delta = 0.122$.

We used the group sequential design framework in \citet{hampson2013group}, which involves interim analyses where a choice is made to stop or continue enrollment; if enrollment is stopped, then one waits until all pipeline participants complete the trial and then  
a hypothesis test (called a decision analysis) is conducted. 
Full details of the design are  in Supplementary Material \ref{appen:gst-multiple-testing}. 

Denote by $K$ the total number of stages. We set $K=5$ stages in our simulated trials. Participants were enrolled at the constant rate of 140 participants per year based on the projection for the enrollment rate in the MISTIE-III trial \citep{MistieIIITrial}.
Given the maximum sample size $\nmax$, the timing of interim analyses is chosen such that for $1\leq k \leq 4$, at the $k$-th interim analysis there are $(k/K)\nmax$ participants with $Y$ observed.

For either estimator the corresponding Wald statistics at different analysis times have (asymptotically) a multivariate normal distribution, where the information at a given analysis time is the reciprocal of the estimator variance. (This  relies on Assumptions 1-3 and holds regardless of the correctness of the outcome regression models used in the adjusted estimator.) The covariance matrix can be estimated, e.g., using the nonparametric bootstrap. To reduce the computational burden in our simulation studies, for any data generating distribution and $\nmax$ we 
 precompute the covariance matrix of our statistics at different analysis times, which is then used to determine stopping boundaries. 
 The efficacy/futility stopping boundaries at each stage are calculated using the error spending approach presented in Supplementary Material \ref{appen:gst-boundary}. We used error spending functions $f(t) = 0.025\min(t^2,1)$ for Type I error  and $g(t) = 0.2\min(t^2,1)$ for Type II error, where $t$ denotes the information time (observed information divided by maximum information).

For each combination of prognostic setting ($\progwl, \progw, \progl, \prognon$) and estimator (unadjusted or adjusted), we use binary search to find the minimum total sample size $\nmax$ required to achieve approximately 80\% power under $\Delta=0.122$ at $0.025$ Type I error using the group sequential design. In practice, this could be approximated from  accruing data by estimating the variance of the adjusted estimator and then using sample size re-estimation to set $\nmax$.
However, we assumed that this was known in order to reduce the computational burden.
For the unadjusted estimator, in every setting we let $\nmax$ be 480. 
For the adjusted estimator, $\nmax$ is set to be 300 under all settings with $\progwl$ and $\progw$, and 480 under all settings with $\progl$ and $\prognon$. We simulated 50,000 trials under each of the eight settings.
The number of accrued participants, error spending functions, and the stopping boundaries for each analysis are listed in Table~\ref{tab:error-boundary} in the Supplementary Material.

\subsection{Simulation Results: Sample Size Reduction from Adjustment} \label{subsec:simulationresults-2}

Table \ref{tab:poweress} lists the simulation-based $\nmax$, Type I error (obtained under $\Delta=0$), power (obtained under $\Delta=0.122$), and the expected sample size (ESS) for each of the eight settings for each estimator. For each setting, the empirical Type I error rate is controlled at 0.025 and the power is approximately 80\%. ESS is calculated as the number of enrolled participants when the trial stops, averaging over 50,000 simulations. The performance of the unadjusted estimator is the same under all prognostic settings and is summarized in the first row.

Compared to the unadjusted estimator, the adjusted estimator  substantially reduces the sample size when the baseline variables are prognostic. In particular, comparing the first row (unadjusted) versus the third row (adjusted $\progw$), the adjusted estimator reduces the maximum sample size from 480 to 300, the expected sample size under $\Delta=0$ from 318 to 227, and the expected sample size under $\Delta=0.122$ from 382 to 260. The sample size reduction due to the prognostic short-term outcome is very small. When neither $W$ nor $L$ is prognostic ($\prognon$), performance of the two estimators is similar.


\begin{table}[htbp]
\caption{The maximum sample size ($\nmax$) and empirical Type I error, power, and expected sample size (ESS) under $\Delta=0$ and $\Delta=0.122$ for each estimator under each prognostic setting. The $\nmax$ was chosen for  each prognostic setting in order to achieve approximately 80\% power under $\Delta=0.122$.}
\label{tab:poweress}
\begin{center}
\begin{tabular}{llrrrrr}
\hline
Estimator  	& Progn. set.	& $\nmax$	& Type I error	& Power	& ESS($\Delta=0$) 	& ESS($\Delta=0.122$) \\
\hline
unadjusted 	& -			    & 480     	& 0.0250			& 0.811  	& 318   		& 382 \\
            &           		&          	&   				&       	&          	&     \\
adjusted    & $\progwl$   	& 300     	& 0.0254 			& 0.791  	& 225       	& 259 \\
           	& $\progw$		& 300 		& 0.0256      	    & 0.805  	& 227       	& 260 \\
           	& $\progl$   	& 480 		& 0.0253			& 0.805  	& 309       	& 375 \\
           	& $\prognon$ 	& 480 		& 0.0248        	& 0.811  	& 321       	& 384 \\
\hline \\
\end{tabular}
\end{center}
\end{table}

\subsection{Simulation Results: Relative Efficiency} \label{subsec:simulationresults}
In this subsection, 
we  simulated trials with no early stopping, i.e.,  each simulated trial always enrolls $\nmax$ participants and continues follow-up until all participants have $Y$ observed. We then look back at each simulated trial and compute what each estimator's value would be at each of the interim and decision analysis times. At any analysis time, the number of participants with $Y$ observed, $L$ observed, and only $(A,W)$ observed, respectively, is fixed; therefore, we can make a direct comparison between the ratio of estimator variances in our simulations and the predicted $\are$ from our formula (which assumes constant $p_l,p_y$). 

In order to make direct comparisons between  estimators at the same sample size, we change (just for this subsection) the setting of $\nmax$. 
The maximum sample size $\nmax$ is set to be 300 under $\progwl$ and $\progw$, and is set to 480 under $\progl$ and $\prognon$. Under each setting the same $\nmax$ is used for both the unadjusted estimator and the adjusted estimator.

Table \ref{tab:re} lists the approximate, asymptotic relative efficiency (ARE) computed by  evaluating (\ref{eq:re-ate}) but with conditional expectations replaced by empirical estimates as discussed in Section \ref{subsec:are-fitted}, and the relative efficiency (RE) from simulation. 
Under most settings, the ARE predicted by the theory is similar to the RE computed from the simulation. There are some discrepancies between these at the earlier analysis times. We think that this is due to the relatively small sample sizes at these analysis times compared to the number of variables adjusted for, which can lead to model overfit and so reduced performance of the adjusted estimator. For example, at interim analysis 1 under $\progwl$ there are only 60 participants with the primary outcome observed but 4 variables get adjusted for in the regression model for $E(Y|W,A,L)$; this violates the rule of thumb that one should have at least 20 observations per term in the regression model. We discuss possible remedies for this in Section~\ref{sec:discussion}.

A comparison between $\progw$ (when only $W$ is prognostic) and $\progl$ (when only $L$ is prognostic) shows that there is a much larger precision gain in the former case. This is despite the fact that $W$ and $L$ have roughly similar prognostic values for $Y$ marginally, as shown in Table~\ref{tab:rsq-mistie}. The reason, as discussed in Sections~\ref{sec:impactbaseline}-\ref{impactcomparison}, is that there is relatively little treatment effect heterogeneity, which means that the baseline variables contribute to precision gains in two ways: allowing adjustment for chance imbalances 
(i) between study arms among the participants with $Y$ observed, and (ii) between the full sample and the subset with $Y$ observed. The short-term outcome can only contribute in the second way, and so may lead to smaller gains.
Under $\progl$, comparing interim analyses to decision analyses shows that $L$ improves estimation precision only at interim analyses (i.e., when there are pipeline participants), which is in line with our theoretical results.

\begin{table}[htbp]
\caption{Comparison of the asymptotic relative efficiency (ARE) predicted by Theorem \ref{cor:re-ate} to the relative efficiency (RE)  from simulated trials, at each analysis time. Top half is under $\Delta=0$; bottom half is under $\Delta=0.122$.
The maximum sample size $\nmax$ is set to be 300 under $\progwl$ and $\progw$, and is set to 480 under $\progl$ and $\prognon$. Under each setting the same $\nmax$ is used for both the unadjusted estimator and the adjusted estimator.
The simulated RE is computed from 50,000 simulated trials. } 
\label{tab:re}
\begin{center}
\resizebox{\linewidth}{!}{
\begin{tabular}{lrrrrrrrrrr}
\hline
 & & \multicolumn{4}{c}{ARE from theory}  & & \multicolumn{4}{c}{RE from Simulation} \\
  \cline{3-6} \cline{8-11}
& & $\progwl$ & $\progw$ & $\progl$ & $\prognon$ & & $\progwl$ & $\progw$ & $\progl$ & $\prognon$ \\ 
  \\
  & & \multicolumn{9}{c}{\textit{under $\Delta=0$ }} \\
 \multirow{4}{*}{\begin{tabular}{l} Interim \\ Analysis \end{tabular}} 
   & 1 & 1.63 & 1.54 & 1.13 & 1.00 && 1.49 & 1.41 & 1.08 & 0.96 \\ 
   & 2 & 1.59 & 1.54 & 1.07 & 1.00 && 1.54 & 1.49 & 1.06 & 0.98 \\ 
   & 3 & 1.58 & 1.54 & 1.05 & 1.00 && 1.54 & 1.51 & 1.04 & 0.99 \\ 
   & 4 & 1.57 & 1.54 & 1.04 & 1.00 && 1.54 & 1.51 & 1.04 & 0.99 \\
  \\
 \multirow{5}{*}{\begin{tabular}{l} Decision \\ Analysis \end{tabular}} 
   & 1 & 1.53 & 1.53 & 1.00 & 1.00 && 1.49 & 1.51 & 0.99 & 0.99 \\ 
   & 2 & 1.53 & 1.53 & 1.00 & 1.00 && 1.50 & 1.53 & 0.99 & 0.99 \\ 
   & 3 & 1.53 & 1.53 & 1.00 & 1.00 && 1.51 & 1.52 & 0.99 & 0.99 \\ 
   & 4 & 1.53 & 1.53 & 1.00 & 1.00 && 1.51 & 1.54 & 1.00 & 1.00 \\ 
   & 5 & 1.53 & 1.53 & 1.00 & 1.00 && 1.51 & 1.54 & 1.00 & 1.00 \\ 
  \\
   & & \multicolumn{9}{c}{\textit{under $\Delta=0.122$ }} \\
 \multirow{4}{*}{\begin{tabular}{l} Interim \\ Analysis \end{tabular}} 
   & 1 & 1.64 & 1.56 & 1.12 & 1.00 && 1.49 & 1.43 & 1.08 & 0.96 \\ 
   & 2 & 1.61 & 1.55 & 1.07 & 1.00 && 1.56 & 1.51 & 1.06 & 0.98 \\ 
   & 3 & 1.59 & 1.55 & 1.05 & 1.00 && 1.57 & 1.53 & 1.04 & 0.99 \\ 
   & 4 & 1.58 & 1.55 & 1.04 & 1.00 && 1.57 & 1.53 & 1.03 & 0.99 \\ 
  \\
 \multirow{5}{*}{\begin{tabular}{l} Decision \\ Analysis \end{tabular}} 
   & 1 & 1.55 & 1.55 & 1.00 & 1.00 && 1.51 & 1.52 & 0.99 & 0.99 \\ 
   & 2 & 1.55 & 1.55 & 1.00 & 1.00 && 1.53 & 1.53 & 0.99 & 0.99 \\ 
   & 3 & 1.55 & 1.55 & 1.00 & 1.00 && 1.53 & 1.53 & 0.99 & 0.99 \\ 
   & 4 & 1.55 & 1.55 & 1.00 & 1.00 && 1.54 & 1.54 & 1.00 & 1.00 \\ 
   & 5 & 1.55 & 1.55 & 1.00 & 1.00 && 1.54 & 1.54 & 1.00 & 1.00 \\ 
   \hline
\end{tabular}
}
\end{center}
\end{table}

\section{Discussion} \label{sec:discussion}

We considered independent
censoring, and  the only source of missing data in our simulations was administrative censoring. 
If outcomes are missing at random \citep{rubin1976MCAR}, then the adjusted estimator may still be  consistent under correct specification of certain working models, while the unadjusted estimator may be biased. 
If dropout is missing not at random \citep{rubin1976MCAR}, then both  estimators may be inconsistent and our formulas will not work. 




The theoretical ARE can be larger than the simulation-based ARE, especially at the early interim analyses. One reason is that in practice there is a finite sample penalty for adjustment (for each degree of freedom in the models fit) that is not reflected in the theoretical calculations, and that dissipates as sample size grows. 
This issue may be mitigated by only scheduling  interim analyses for times when at least 50\% of the participants have primary outcomes observed, and by constraining the number of variables adjusted for such that one has at least 20 participants with primary outcomes observed per variable.

The sample size reductions due to improved precision from adjustment came primarily from selecting smaller $\nmax$. Since the prognostic value of variables is typically not precisely known before the trial starts, one may use a preplanned rule
 for sample size re-estimation to set $\nmax$ based on accruing data. First, the original, maximum sample size is set conservatively assuming no gains from adjustment. During the trial  
the asymptotic variance of the adjusted estimator  is estimated and $\nmax$ is set to be the  sample size for which the desired power is achieved  when $\Delta$ equals the a priori specified, clinically meaningful, minimum treatment effect. The sample size re-estimation should only be conducted after a sufficient number of participants have had their primary outcomes measured, e.g., when roughly 50\% of the originally planned sample size have this measured. 

This approach is only feasible if the delay time $d_Y$ is not too long compared to the enrollment rate (otherwise the originally planned sample size will all be enrolled before a sufficient number of primary outcomes are observed). This was feasible in the MISTIE trial context with $d_Y=180$ days and enrollment rate 140 per year, since sample size re-estimation could be done at 2.1 years. At this time, 
 230 (48\%) of the originally planned sample size (480)  have primary outcomes observed. So restricting to approximately 300 participants (the required $\nmax$ when $W$ is as prognostic as in the MISTIE-II data) is possible.

A limitation of our simulation study is that we precomputed $\nmax$ for each setting, in order to save computation time. In practice, this would be calculated as a function of the estimated variance using accruing data in the trial. A future area of research is to run simulation studies to evaluate the resulting impact on power and sample size.


\textsf{R} code \citep{R} for our simulations can be downloaded at \url{https://github.com/tqian/gst_tmle}.

\section*{Acknowledgments}

This work was supported by the Patient-Centered Outcomes Research Institute (ME-1306-03198), the US Food and Drug Administration (HHSF223201400113C), and NIH grant UL1TR001079. 
This work is solely the responsibility of the authors and does not represent the views of the above agencies.
We thank Mary Joy Argo for helpful comments.

\bibliographystyle{biom}
\bibliography{MISTIE_TMLE_short-term}

\newpage

\appendix

\section*{Supplementary Material}

\numberwithin{equation}{section}
\numberwithin{theorem}{section}
\numberwithin{assumption}{section}
\numberwithin{remark}{section}
\numberwithin{table}{section}
\numberwithin{figure}{section}

In Section \ref{appen:are-sample-size-reduction}, we discuss the relationship between asymptotic relative efficiency and sample size reduction for a Wald test of a single stage trial.
In Section \ref{appen:gst}, we present the group sequential design used in the simulation studies in Section \ref{sec:simulation}. In Section \ref{appen:dgm}, we specify the data generating distribution used in the simulation studies in Section \ref{sec:simulation}. Intuition for the impact of treatment effect heterogeneity on precision gains is given in Section~\ref{appen:adjustmentintuition}. In Section \ref{appen:theory-predictive-prognostic}, we provide additional theoretical results regarding the precision gain from adjusting for prognostic baseline variables, when the baseline variable is either purely predictive or purely prognostic. In Section \ref{sec:theory-singlearm}, we provide theoretical results regarding the precision gain from adjusting for prognostic baseline variables and the short-term outcome, when the parameter is $E(Y \mid A=a)$ for each $a\in\{0,1\}$ (rather than the average treatment effect). 
Simulations to evaluate the theory in Section \ref{sec:theory-singlearm} are provided in Section \ref{appen:re-h1}. Section \ref{appen:sec:proof} includes proofs of the results in the main paper as well as results in Section \ref{appen:theory-predictive-prognostic} and Section \ref{sec:theory-singlearm}. Auxiliary lemmas that are used in the proofs in Section \ref{appen:sec:proof} are themselves proven in Section \ref{appen:proof-aux}.

\section{Relationship between asymptotic relative efficiency and sample size reduction for a Wald test of a single stage trial}
\label{appen:are-sample-size-reduction}

Consider one-sided Wald tests of the null hypothesis $H_0:\Delta \leq 0$ versus the alternative $H_1:\Delta > 0$, using RAL estimators $\widehat{\Delta}_1$ and $\widehat{\Delta}_2$, respectively (where the test statistics are the estimators divided by their standard errors). The asymptotic relative efficiency can be interpreted as the limit as sample size goes to infinity of the inverse of the ratio of the required sample sizes for the two estimators to achieve a given power at local alternatives \citep[Section 8.2]{vdv}. For example, $\are(\widehat{\Delta}_1, \widehat{\Delta}_2) = 1.2$ means that 
by using $\widehat\Delta_1$ instead of $\widehat\Delta_2$, the required sample size is reduced by $1 - 1 / 1.2 \approx 17\%$ asymptotically. 

\section{Full Description of Group Sequential Design Used in Simulation} \label{appen:gst}
Components of the group sequential design used in simulation are presented in the following order: the multiple testing procedure (Section \ref{appen:gst-multiple-testing}), the computation of the test boundaries (Section \ref{appen:gst-boundary}), the sample size at each analysis for the simulated trials (Section \ref{appen:sample-size-each-analysis}), and the value of Type I error and Type II error spent and the testing boundary at each analysis (Section \ref{appen:testing-boundary}).
\subsection{Multiple Testing Procedure} \label{appen:gst-multiple-testing}

In the simulation studies in Section \ref{sec:simulation} we use the  group sequential test from \citet{hampson2013group} with $K$ stages. Given $\alpha, \beta \in (0,1)$, the design goal is to control the Type I error rate at level $\alpha$ and have power $1-\beta$ at alternative $\Delta=\delta>0$. Such a group sequential test can terminate enrollment at an interim analysis, and if such early stopping happens follow-up continues until all pipeline participants have $Y$ observed before conducting a decision analysis to reject or accept $H_{0}$. For each stage $k$, denote by $S_{k}$ and $\tilde{S}_{k}$ the test statistics at the $k$th interim analysis (where the decision to stop or continue enrollment occurs) and the $k$th decision analysis (where the hypothesis test is conducted), respectively; let  $u_{k}$ and $l_{k}$ denote the efficacy and futility boundaries for terminating enrollment at interim analysis $k$, and let $c_{k}$ denote the critical value for the corresponding decision analysis. These are used in the group sequential testing procedure below, reproduced from \citet[Figure 1]{hampson2013group}:

\bigskip

\makebox[\textwidth]{
$\begin{cases}
\mbox{At interim analysis }k=1,\ldots,K-2,\\
\quad\mbox{if }S_{k}\leq l_{k}\mbox{ or }S_{k}\geq u_{k} & \mbox{stop recruitment and proceed to decision analysis }k;\\
\quad\mbox{otherwise} & \mbox{continue recruitment and proceed to interim analysis \ensuremath{k+1}.}\\
\mbox{At interim analysis }K-1,\\
\quad\mbox{if }S_{K-1}\leq l_{K-1}\mbox{ or }S_{K-1}\geq u_{K-1} & \mbox{stop recruitment and proceed to decision analysis }K-1;\\
\quad\mbox{otherwise} & \mbox{complete recruitment and proceed to decision analysis \ensuremath{K}.}\\
\mbox{At decision analysis }k=1,\ldots,K,\\
\quad\mbox{if }\tilde{S}_{k}\geq c_{k} & \mbox{reject }H_{0};\\
\quad\mbox{if }\tilde{S}_{k}<c_{k} & \mbox{accept }H_{0}.
\end{cases}$
}

\subsection{Computation of Test Boundaries} \label{appen:gst-boundary}
Following \citet{hampson2013group}, consider a Type I error spending function $f(\cdot)$ and Type II error spending function $g(\cdot)$, which are non-decreasing with $f\left(0\right)=g\left(0\right)=0$ and $f\left(t\right)=\alpha$ and $g\left(t\right)=\beta$ for $t\geq1$. The maximum information level ${\cal I}_{\max}$ is chosen depending on the power goal and the error spending functions. Denote by ${\cal I}_{k}$ and $\tilde{{\cal I}}_{k}$ the information levels at the $k$-th interim analysis and decision analysis, respectively. Denote by ${\cal C}_{k}=\left(l_{k},u_{k}\right)$ the critical region at interim analysis $k$, $1\leq k \leq K$. The test boundaries $u_k$, $l_k$, and $c_k$ are calculated by (12)-(15) in \citet[Section 4.1.1]{hampson2013group}, and we paraphrase as follows. Let $u_1$ and $l_1$ be the solutions to
\begin{equation}
P(S_1\geq u_1;\Delta=0)=f({\cal I}_1/{\cal I}_{\max})\qquad\mbox{and}\qquad P(S_1\leq l_1;\Delta=\delta)=g({\cal I}_1/{\cal I}_{\max}). \nonumber
\end{equation}
For $2\leq k\leq K-1$, $u_k$ is the solution to
\begin{equation}
P(S_1\in{\cal C}_1,\ldots,S_{k-1}\in{\cal C}_{k-1},S_k\geq u_k;\Delta=0)=f({\cal I}_k/{\cal I}_{\max})-f({\cal I}_{k-1}/{\cal I}_{\max}),  \label{eq:err_bdry_effic}
\end{equation}
and $l_k$ is the solution to
\begin{equation}
P(S_1\in{\cal C}_1,\ldots,S_{k-1}\in{\cal C}_{k-1},S_k\leq l_k;\Delta=\delta)=g({\cal I}_k/{\cal I}_{\max})-g({\cal I}_{k-1}/{\cal I}_{\max}). \label{eq:err_bdry_fut}
\end{equation}
For $1\leq k\leq K-1$, the critical value $c_k$ is the solution to
\begin{align}
P(S_1\in{\cal C}_1,\ldots,S_{k-1}\in{\cal C}_{k-1}, & S_k\geq u_k,\tilde{S}_k<c_k;\Delta=0) \nonumber \\
& = P(S_1\in{\cal C}_1,\ldots,S_{k-1}\in{\cal C}_{k-1},S_k\leq l_k,\tilde{S}_k\geq c_k;\Delta=0).  \nonumber 
\end{align}
And the critical value $c_K$ for the last stage is the solution to
\begin{equation}
P(S_1\in{\cal C}_1,\ldots,S_{K-1}\in{\cal C}_{K-1},\tilde{S}_K\geq c_K;\Delta=0)=\alpha-f({\cal I}_{K-1}/{\cal I}_{\max}).  \nonumber 
\end{equation}

\subsection{Sample Size at Each Analysis}
\label{appen:sample-size-each-analysis}

Table \ref{tab:sample-size} lists the sample size and analysis timing of the group sequential designs with $\nmax=480$ and $\nmax=300$ used in Section \ref{subsec:simulationresults}. For $1\leq k \leq 4$, Analysis $k$ indicates interim analysis at stage $k$ and $k^*$ indicates the corresponding decision analysis if enrollment is early stopped at that stage. Analysis $5^*$ indicates the final decision analysis. There is not any  interim analysis for the final stage. Fully observed participants are those with $W,L,Y$ observed; partially observed participants are those with $W,L$ but not $Y$ observed. The three groups of participants listed in Table \ref{tab:sample-size} are inclusive of all enrollees and mutually exclusive.

\begin{table}[htbp]
\caption{Analysis time and sample size at each interim and decision analysis for group sequential designs with $\nmax=480$ and $\nmax=300$. For $1\leq k \leq 4$, Analysis $k$ indicates interim analysis and $k^*$ indicates the corresponding decision analysis if enrollment is early stopped. Analysis $5^*$ indicates the final decision analysis.}
\label{tab:sample-size}
\begin{center}
\begin{tabular}{llrrrrrrrrr}
\hline
Analysis                        & & 1     & $1^*$   & 2     & $2^*$   & 3     & $3^*$   & 4     & $4^*$   & $5^*$   \\
\hline
                          & & \multicolumn{9}{c}{Design with  $\nmax = 480$}  \\
Time (year)                           & & 1.2 & 1.7   & 1.9   & 2.4   & 2.6   & 3.0   & 3.2   & 3.7   & 3.9   \\
\# Fully observed ($W,L,Y$)         & & 96    & 165  & 192  & 261   & 288  & 357  & 384  & 453  & 480  \\
\# Partially observed ($W,L$) only      & & 57    & 0     & 57    & 0     & 57    & 0     & 57    & 0     & 0     \\
\# Pipeline with only $W$ observed & & 12     & 0     & 12    & 0     & 12    & 0     & 12    & 0     & 0     \\
\\
                          & &\multicolumn{9}{c}{Design with  $\nmax = 300$}             \\
Time (year)                           & & 0.9 & 1.4   & 1.4   & 1.8   & 1.8   & 2.3   & 2.2   & 2.6   & 2.6   \\
\# Fully observed ($W,L,Y$)         & & 60    & 129  & 120  & 189   & 180 & 249  & 240  & 300  & 300  \\
\# Partially observed ($W,L$) only      & & 57    & 0     & 57    & 0     & 57    & 0     & 57    & 0     & 0     \\
\# Pipeline with only $W$ observed & & 12     & 0     & 12    & 0     & 12    & 0     & 3     & 0     & 0     \\
\hline \\
\end{tabular}
\end{center}
\end{table}

\subsection{Type I Error and Type II Error Spent and Testing Boundary}
\label{appen:testing-boundary}

Table \ref{tab:error-boundary} lists the Type I error per stage $f({\cal I}_k/{\cal I}_{\max})-f({\cal I}_{k-1}/{\cal I}_{\max})$, Type II error per stage $g({\cal I}_k/{\cal I}_{\max})-g({\cal I}_{k-1}/{\cal I}_{\max})$, and the testing boundaries at each stage under different settings used in the simulation studies. Given any prognostic setting $\progwl, \progw, \progl, \prognon$ and estimator pair, the designs are the same for $\Delta=0$ and $\Delta=0.122$.

\begin{table}[htbp]
\caption{Type I error per stage $f({\cal I}_k/{\cal I}_{\max})-f({\cal I}_{k-1}/{\cal I}_{\max})$, Type II error per stage $g({\cal I}_k/{\cal I}_{\max})-g({\cal I}_{k-1}/{\cal I}_{\max})$, and boundaries for the designs in the different settings used in Section \ref{subsec:simulationresults-2}.}
\label{tab:error-boundary}
\begin{center}
\begin{small}
\begin{tabular}{lrrrrr}
\hline
Analysis $(k)$                          & 1     & 2     & 3     & 4     & 5      \\
\hline
                                    & \multicolumn{5}{c}{TMLE under $\progwl$, $\nmax = 300$} \\
Type I error per stage $\times 10^{-3}$      & 0.9    & 3.2   & 5.1   & 7.2   & 8.6   \\
Type II error per stage  $\times 10^{-3}$       & 7.5   & 25.4    & 41.1    & 57.3    & 68.6    \\
Efficacy boundary at interim analysis ($u_k$) & 3.11    & 2.71    & 2.47    & 2.27    & -     \\
Futility boundary at interim analysis ($l_k$)   & -1.20 & -0.08   & 0.70    & 1.36    & -     \\
Critical value at decision analysis ($c_k$)     & 1.30    & 1.54    & 1.74    & 1.91    & 2.07    \\
\\
                                    & \multicolumn{5}{c}{TMLE under $\progw$, $\nmax = 300$}  \\
Type I error per stage $\times 10^{-3}$     & 0.8   & 2.9   & 5.0   & 7.0   & 9.3   \\
Type II error per stage  $\times 10^{-3}$     & 6.5   & 23.4    & 40.1    & 55.8    & 74.1    \\
Efficacy boundary at interim analysis ($u_k$) & 3.16    & 2.73    & 2.48    & 2.28    & -     \\
Futility boundary at interim analysis ($l_k$)   & -1.26 & -0.14   & 0.66    & 1.33    & -     \\
Critical value at decision analysis ($c_k$)     & 1.32    & 1.55    & 1.75    & 1.91    & 2.06    \\
\\
                                    & \multicolumn{5}{c}{TMLE under $\progl$, $\nmax = 480$}  \\
Type I error per stage $\times 10^{-3}$     & 1.2   & 3.2   & 5.3   & 7.5   & 7.8   \\
Type II error per stage  $\times 10^{-3}$     & 9.3   & 25.8    & 42.6    & 60.2    & 62.1    \\
Efficacy boundary at interim analysis ($u_k$) & 3.05    & 2.68    & 2.44    & 2.24    & -     \\
Futility boundary at interim analysis ($l_k$)   & -0.99 & 0.04    & 0.80    & 1.47    & -     \\
Critical value at decision analysis ($c_k$)     & 1.24    & 1.50    & 1.72    & 1.91    & 2.06    \\
\\
                                    & \multicolumn{5}{c}{TMLE under $\prognon$, $\nmax = 480$}  \\
Type I error per stage $\times 10^{-3}$     & 0.9   & 2.9   & 5.1   & 6.9   & 9.3   \\
Type II error per stage  $\times 10^{-3}$     & 7.3   & 22.9    & 40.6    & 54.9    & 74.3    \\
Efficacy boundary at interim analysis ($u_k$) & 3.12    & 2.73    & 2.47    & 2.28    & -     \\
Futility boundary at interim analysis ($l_k$)   & -1.16 & -0.08   & 0.72    & 1.37    & -     \\
Critical value at decision analysis ($c_k$)     & 1.24    & 1.50    & 1.72    & 1.91    & 2.04    \\
\\
                                                          & \multicolumn{5}{c}{Unadjusted estimator, $\nmax = 480$}  \\
Type I error per stage $\times 10^{-3}$     & 1.0   & 3.0   & 4.9   & 6.9   & 9.2   \\
Type II error per stage  $\times 10^{-3}$     & 7.8   & 24.1    & 39.6    & 54.9    & 73.6    \\
Efficacy boundary at interim analysis ($u_k$) & 3.10    & 2.72    & 2.48    & 2.29    & -     \\
Futility boundary at interim analysis ($l_k$)   & -1.37 & -0.43 & 0.24    & 0.82    & -     \\
Critical value at decision analysis ($c_k$)     & 1.18    & 1.36    & 1.53    & 1.67    & 2.11    \\
\hline \\
\end{tabular}
\end{small}
\end{center}
\end{table}

\section{Resampling-based Algorithm to Simulate Trial Data} \label{appen:dgm}

For the data generating distribution to have the two properties in Section \ref{subsec:dgm}, 
we design the algorithm to generate a simulated trial of $n$ independent and identically distributed samples from the 100 participants in MISTIE-II dataset. 
Recall that the three properties are:
\begin{enumerate}
  \item[(i)] The treatment assignment is independent of baseline variables.
  \item[(ii)] The relative efficiency between the adjusted estimator and the unadjusted estimator in the simulated data is similar to that calculated directly from the MISTIE-II dataset.
\end{enumerate}
Briefly, the algorithm ensures property (i) by adding a ``twin'' with identical baseline variables and opposite treatment assignment to each participant in the MISTIE-II dataset. The variables $L^\text{full}$ and $Y$ for each ``twin'' are then generated using regression model fits. The original MISTIE-II data and the set of ``twins''  results in a 200 participant data set. This is done once, before any of our simulations are conducted. To generate each simulated data set, we resample participant vectors from this 200 participant data set with replacement and then make modifications to the replicate data set depending on the desired simulation setting. Details of the algorithm are given below.

\medskip

\textit{Step 1: Construct a set of 100 pairs of ``twins''.} We start with the 100-participant MISTIE-II dataset, and we augment the data with a hypothetical ``twin'' for each participant. A ``twin'' has identical baseline variables as the original participant, but opposite treatment assignment. We fit logistic regression models for $L^{(1)}$ on $(W^\text{full},A)$, for $L^{(2)}$ on $(W^\text{full},A,L^{(1)})$, and for $Y$ on $(W^\text{full},A,L^{(1)},L^{(2)})$, using the original 100 participants in MISTIE-II dataset. The preliminary $L_i$ and $Y_i$ of each newly added twin are then predicted based on these logistic regression models by rounding the predicted success probability to be $0$ or $1$. The indicator of whether a participant is an original participant in the MISTIE-II data set or a hypothetical twin is included as a variable in this augmented data set of 200 participants, which will be used in Step 2. Step 1 is only done once at the beginning of the simulation.

\textit{Step 2: Sample $n$ participants from the augmented data set with 200 participants and calibrate the treatment effect on the sampled participants.} 
We sample $n$ participants uniformly with replacement from the augmented dataset constructed in Step 1. Then, for each participant in the $n$ participants whose ``hypothetical twin indicator'' is true, with probability $0.03$ we reset its $Y_i$ to equal $A_i$. This resetting step increases the treatment effect of the augmented data to $0.122$, matching that of the original data.

\textit{Step 3: Calibrate relative efficiency between the adjusted estimator and the unadjusted estimator.} For $a\in\{ 0,1\}$, we empirically estimate the marginal distribution $p_{Y;a}=P(Y=1\mid A=a)$ using the corresponding sample proportions in the original data. Then, for each participant in the $n$ participants whose ``hypothetical twin indicator'' is true, with probability $0.164$ we reset $Y_i$ by a realization of an independent Bernoulli random draw with success probability $p_{Y;A_i}$. This resetting step adds random noise to reduce the prognostic value in $W$ and $L$, so that the relative efficiency between the adjusted estimator and the unadjusted estimator for the simulated data is comparable to the estimated $\are$ from the original MISTIE-II dataset.

\bigskip

We then make modifications that are described next to assess validity of the theory under various settings. We define eight data generating distributions, called ``settings'', each based on modifying the aforementioned data generating process. The goal is to consider four situations involving $W$ and/or $L$ being prognostic or not for $Y$; these are denoted by $\itprogwl$ (both prognostic),  $\itprogw$ (only $W$ prognostic), $\itprogl$ (only $L$ prognostic),  and  $\itprognon$ (neither prognostic).  For each of these four situations, we consider the following two average treatment effects:  $\Delta=0$ (no effect) or $\Delta=0.122$ (benefit). 

We first describe how we generated simulated trials with $\Delta=0.122$. 
The aforementioned data generating algorithm corresponds to $\itprogwl$. To create data generating distributions with $\itprogw$, we first generate each participant's data as in $\itprogwl$ and then replace $L$ by an independent draw from the marginal distribution in the MISTIE-II data. Similarly, in order to make only $L$ prognostic for $Y$ ($\itprogl$), we first generate each participant's data as in $\itprogwl$ and then replace $W$ by an independent draw from the marginal distribution in the MISTIE-II data. Lastly, to make neither $W$ nor $L$ prognostic, we first generate each participant's data as in $\itprogwl$ and then replace each of $W$ and $L$ by independent draws from the corresponding marginal distributions in the MISTIE-II data, respectively.

To generate participants under $\Delta=0$, the data generating process above is followed, and then each study arm assignment $A$ is replaced by an independent Bernoulli draw with probability $1/2$. Under the data generating process with $\Delta=0.122$, there is slight treatment effect heterogeneity $\gamma \approx 0.01$. By construction, under $\Delta=0$ there is no treatment effect heterogeneity ($\gamma = 0$).

\section{Intuition for Impact of Treatment Effect Heterogeneity on Precision Gains} \label{appen:adjustmentintuition}

To give intuition for the phenomena in Figure~\ref{fig:are-w-ate} in Section~\ref{sec:impactbaseline}, for the remainder of this appendix we fix $R^2_W$ and consider the simple case of a single, binary-valued $W$ representing being less than 65 years old at baseline, and primary outcome $Y$  being the indicator of having mRS at most 3 (a good outcome) at 180 days.
 We focus on the stratum $W=1$ and assume that 
the under 65 subset of the study population would have better outcomes on average than the overall study population if all were assigned to the control arm, i.e.,  $E(Y|A=0,W=1)-E(Y|A=0)>0$. 
We next consider $p_y=1$ and the extreme of maximum treatment effect heterogeneity (where there are no efficiency gains from adjusting for $W$). (The opposite extreme of $p_y=1$ and no treatment effect heterogeneity is discussed in detail in Section~\ref{sec:impactbaseline}.)

Consider $p_y=1$ and with maximum treatment effect heterogeneity ($\gamma = 2R^2_W$), which occurs when $E(Y|A=1,W)+E(Y|A=0,W)$ equals the constant $E(Y|A=1)+E(Y|A=0)$. It follows from the assumption above that those under 65 have better outcomes on average than the overall study population in the control arm ($E(Y|A=0,W=1)-E(Y|A=0)>0$), that those under 65 have worse outcomes on average than the overall study population under assignment to the treatment arm (and by precisely the same magnitude), which follows since $E(Y|A=1,W=1)-E(Y|A=1)= -\{E(Y|A=0,W=1)-E(Y|A=0)\}$. If by chance there are proportionally more under 65 participants assigned to the treatment arm than the control arm, then the unadjusted estimator of the mean outcome in each arm is biased downward, and by the same amount, conditional on the chance imbalance. Since this bias cancels out when estimating $\Delta$, an efficient estimator of $\Delta$ makes no adjustment. Alternatively, if by chance there are proportionally fewer under 65 participants assigned to the treatment arm, then by a symmetric argument an efficient estimator of $\Delta$ makes no adjustment. The overall impact is no adjustment due to chance imbalances in the stratum $W=1$, which means that no improvement in variance is made compared to the unadjusted estimator. This is, intuitively, why there is no efficiency gain under maximum  treatment effect heterogeneity ($\gamma=2R^2_W$) when there are no pipeline participants ($p_y=1$).

We next give intuition for why one needs $\gamma>0$ in order for pipeline participants to contribute useful information toward estimating $\Delta$ (when $p_y < 1$). Consider the case where those in stratum $W=1$ benefit more than the overall population, i.e.,  $E(Y|A=1,W=1)-E(Y|A=0,W=1)>\Delta$; this is only possible if $\gamma>0$.
 Assume that adjustment for chance imbalance in the stratum $W=1$ between arms has already been done to the unadjusted estimator (using only those with $Y$ observed), as described  above. 
Consider the proportion of participants having $W=1$ among all participants and also among the subset of participants with $Y$ observed. If the former proportion is larger than the latter, 
then the stratum $W=1$ (of participants who benefit more than average) is underrepresented among those with outcomes observed. 
 An efficient estimator adjusts the estimate of $\Delta$ upward to compensate. 
 Alternatively, if the former proportion is smaller than the latter, by a symmetric argument an efficient estimator adjusts the estimate of $\Delta$ downward. In this way, variance due to imbalance in the proportion with $W=1$ between the overall sample (including pipeline participants) and those with $Y$ observed is removed. 
 The only cases where adjusting for the aforementioned imbalance at each level of $W$ has no impact are when $E(Y|A=1,W)-E(Y|A=0,W) = \Delta$, i.e., zero treatment effect heterogeneity, or $p_y=1$ (no pipeline participants).

\section{Additional Results: When Baseline Variable is Purely Predictive or Purely Prognostic.}
\label{appen:theory-predictive-prognostic}

The following theorem examines the impact of treatment effect heterogeneity on precision gain from adjusting for prognostic baseline variables under two extreme cases. Define $\are(\eff, \unadj)$ to be the asymptotic relative efficiency in (\ref{eq:re-ate}). More generally, $\are(\widehat{\Delta}_1, \widehat{\Delta}_2)$ denotes the asymptotic relative efficiency of the estimators $\widehat{\Delta}_1, \widehat{\Delta}_2$.

\begin{corollary} 
\label{cor:prognostic-predictive}
Suppose Assumptions \ref{assump:randomization}, \ref{assump:independent-censoring}, and \ref{assump:monotone-censoring} hold. In addition, assume $R^2_{L \mid W} = 0$. For the estimand $E(Y\mid A=1) - E(Y\mid A=0)$, we have the following result regarding the asymptotic relative efficiency between an efficient RAL estimator and the unadjusted estimator.
\begin{enumerate}
\item[(i)] (Maximal treatment effect heterogeneity.) 

If $\var\{E_1(Y \mid W) - E_0(Y \mid W)\} > 0$ and $\var\{E(Y \mid W)\} = 0$, then $\gamma = 2 R_W^2$ and 
\begin{equation}
\are(\eff, \unadj) = \frac{1}{1 - (1-p_y) R^2_W}. \label{eq:rsq-onlyW-ate-predictive}
\end{equation}

\item[(ii)] (No treatment effect heterogeneity.) 

If $\var\{E_1(Y \mid W) - E_0(Y \mid W)\} = 0$ and $\var\{E(Y \mid W)\} > 0$, then $\gamma = 0$ and 
\begin{equation}
\are(\eff, \unadj) = \frac{1}{1 - R^2_W}. \label{eq:rsq-onlyW-ate-prognostic}
\end{equation}

\end{enumerate}
\end{corollary}

\bigskip
We call a baseline variable $W$ \textit{purely predictive} if it satisfies the conditions in Corollary \ref{cor:prognostic-predictive}(i), as it predicts the conditional average treatment effect (i.e., it identifies which strata of $W$ benefit from the treatment on average), but it does not explain the variance in $Y$ marginally (i.e., margining out $A$). For such a baseline variable, $\are = 1$ in (\ref{eq:rsq-onlyW-ate-predictive}) when $p_y=1$. This means that adjusting for a purely predictive baseline variable results in no precision gain when there are no missing primary outcomes. An extreme case of this, for illustration only, is the following data generating distribution where the baseline variable is perfectly correlated with the outcome within each arm but still contributes nothing to the precision gain: $W,A$ are independently distributed as $\text{Bernoulli}(0.5)$ and $Y = \indic(A=W)$, where $\indic(\cdot)$ is the indicator function. In this case it is straightforward to verify that $R^2_W = 1$, but adjusting for $W$ would not impact asymptotic variance compared to the unadjusted estimator if all participants have $Y$ observed $(p_y=1)$ since $\gamma=2$. Though we don't expect to encounter a purely predictive variable in practice, the above thought experiment shows how $W$ could explain treatment effect heterogeneity (which is useful on its own) while being useless for improving precision in estimating the average treatment effect.

We call a baseline variable $W$ \textit{purely prognostic} if it satisfies the conditions in Corollary \ref{cor:prognostic-predictive}(ii), as the treatment effect is constant across strata of $W$, and $W$ explains some of the marginal variance in $Y$. For such a baseline variable, the $\are$ in (\ref{eq:rsq-onlyW-ate-prognostic}) does not depend on $p_y$. An intuition for  why this holds, as discussed in Section~\ref{sec:impactbaseline}, is that adjusting for such a baseline variable reduces variance by correcting for chance imbalances between arms among those who have $Y$ observed, while the baseline variables for the pipeline participants contain no information about $\Delta$ since $\gamma=0$ (and we are ignoring $L$ here).

\section{Precision Gain When Estimating the Treatment Specific Mean $E(Y|A=a)$} \label{sec:theory-singlearm}


In this section we focus on estimating the population mean outcome under assignment to a single arm (called the treatment specific mean): $E(Y\mid A=a)$ for each  $a\in\{0,1\}$ separately. The following theorem gives the lower bound on the asymptotic variance for all regular asymptotically linear (RAL) estimators of $E(Y\mid A=a)$.

\begin{theorem} 
\label{thm:avar}

Denote by $p_a = P(A=a)$ for $a\in\{0, 1\}$. Assume $p_y > 0$ and $0<p_a<1$. Under Assumptions \ref{assump:randomization}, \ref{assump:independent-censoring}, and \ref{assump:monotone-censoring} (except allowing $p_a \neq 1/2$) the semiparametric lower bound on the asymptotic variance of all RAL estimators for $E(Y\mid A=a)$ is
\begin{equation}
\var_a\{E_a(Y \mid W)\} + \frac{1}{p_a p_l} \var_a\{ E_a(Y \mid L,W) - E_a(Y \mid W) \} + \frac{1}{p_a p_y} \var_a\{ Y - E_a(Y \mid L,W) \}. \label{eq:thm:var-bound}
\end{equation}
\end{theorem}

Analogous to the $R$-squared quantities defined in Section \ref{sec:theory-ate}, we define new $R$-squared quantities relevant to estimating the population mean of the primary outcome $Y$ in each arm separately.
For each $a \in \{0,1\}$, define the proportion of variance in $Y$ explained by $W$ in arm $a$ as 
$
R^2_{W;a} = \var\{E_a(Y \mid W)\} / \var_a(Y)
$;
the proportion of additional variance in $Y$ explained by $L$ after accounting for $W$ in arm $a$ as
$
R^2_{L \mid W;a} = \var_a\{E_a(Y \mid L, W) - E_a(Y \mid W)\} / \var_a(Y)
$.
We also will refer to the proportion of variance in $Y$ explained by $L$ alone in arm $a$, defined as  $R^2_{L;a} = \var_a\{E_a(Y \mid L)\} / \var_a(Y)$.



The following result gives the $\are$ between an efficient RAL estimator and the unadjusted estimator for $E(Y\mid A=a)$.

\begin{theorem} 
\label{cor:re}

Assume all conditions in Theorem \ref{thm:avar} hold. For each arm $a\in\{0,1\}$,  the asymptotic relative efficiency between an efficient RAL estimator and the unadjusted estimator of   
the treatment specific mean $E(Y\mid A=a)$  is 
\begin{equation}
\frac{1}{1 - (1 - p_a p_y) R^2_{W;a} - (1 - p_y / p_l) R^2_{L\mid W;a}}. \label{eq:re}
\end{equation}

\end{theorem}

\begin{figure}[htbp]
\caption{Asymptotic relative efficiency between an efficient RAL estimator and the unadjusted estimator for estimating $E(Y\mid A=a)$. Plot (a) is when only the baseline variable $W$ is prognostic; plot (b) is when only the short-term outcome $L$ is prognostic.}
\label{fig:are}
\begin{center}
\begin{subfigure}{.5\textwidth}
  \centering
  \includegraphics[width=\linewidth]{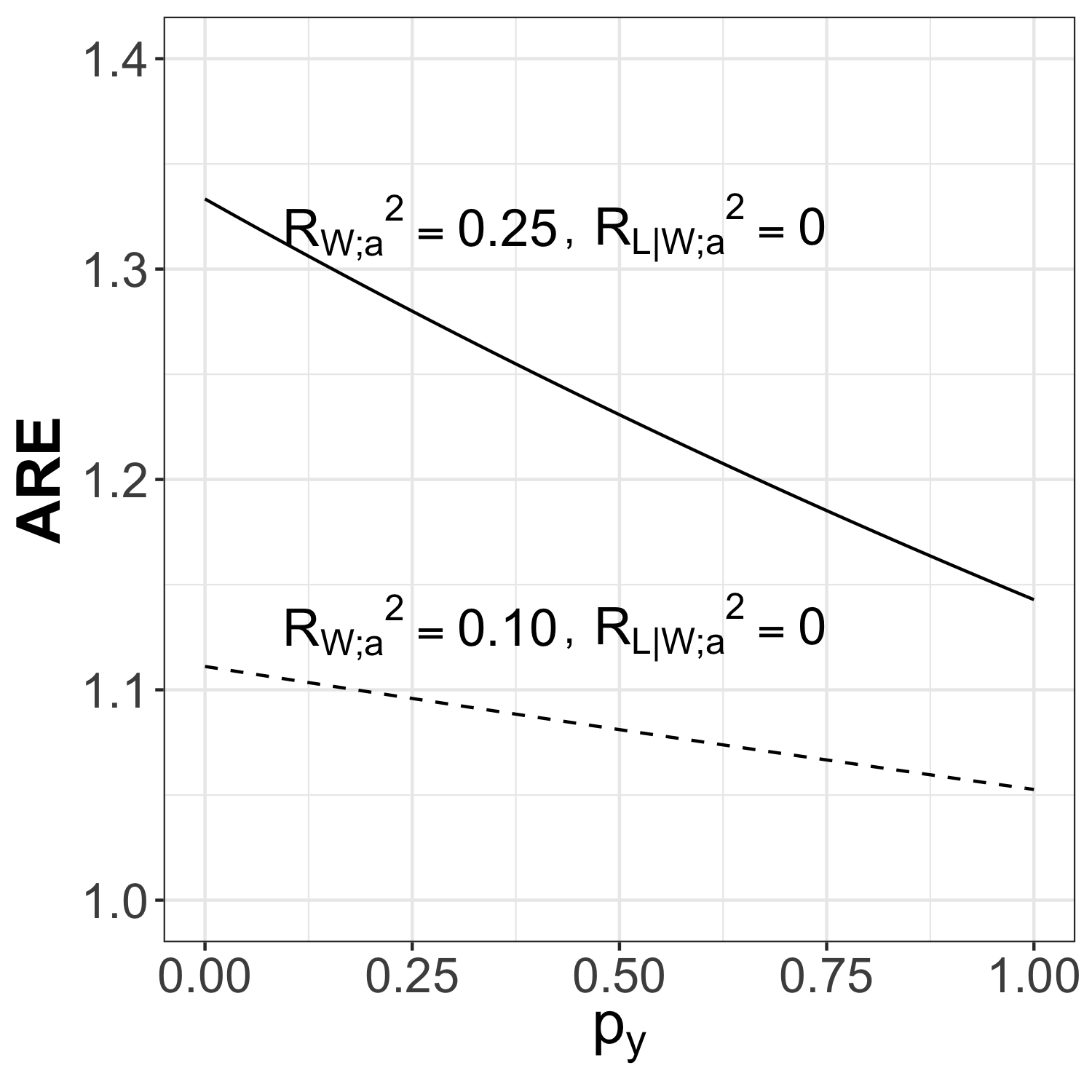}
  \caption{}
  \label{fig:are-w}
\end{subfigure}%
\begin{subfigure}{.5\textwidth}
  \centering
  \includegraphics[width=\linewidth]{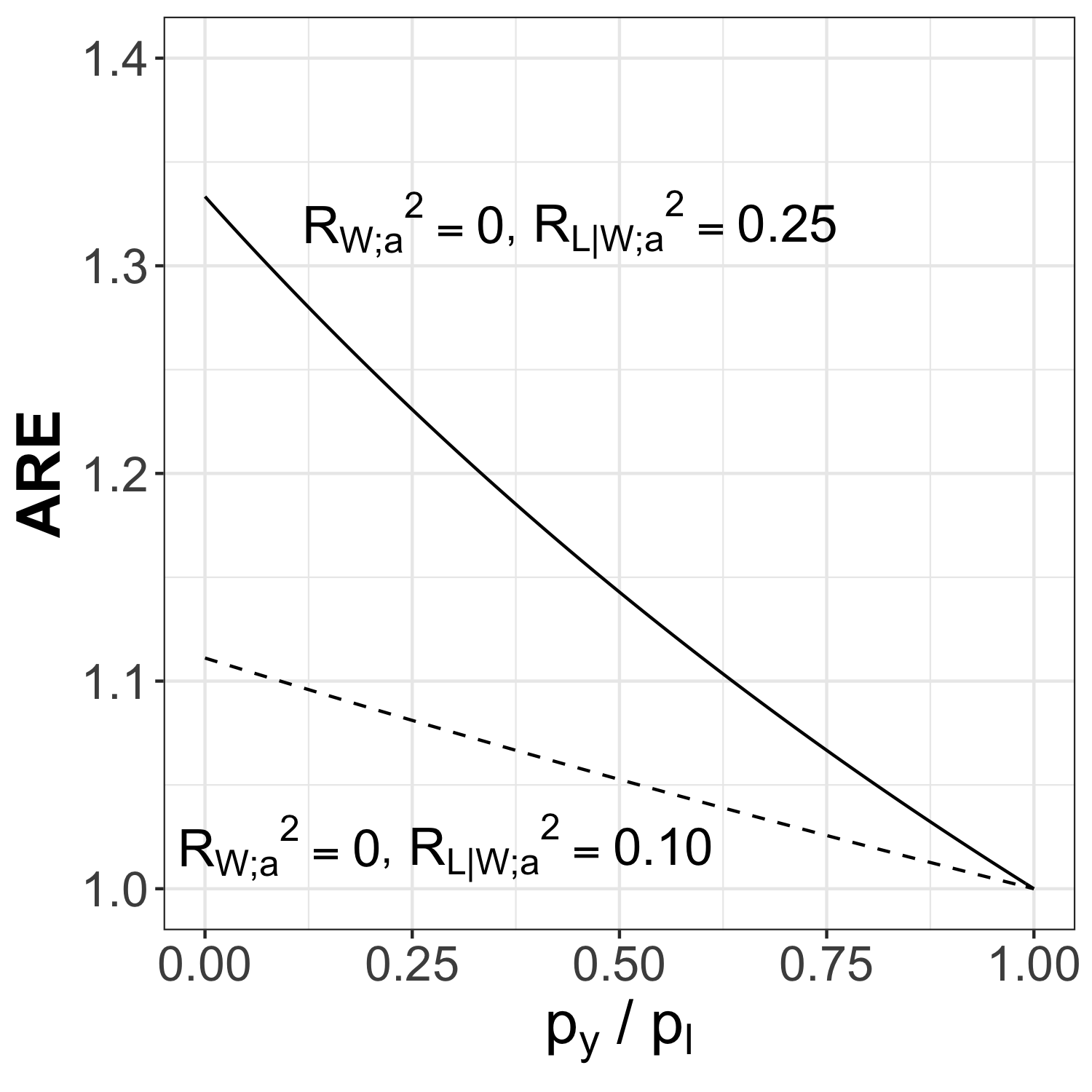}
  \caption{}
  \label{fig:are-l}
\end{subfigure}
\end{center}
\end{figure}

When only $W$ is prognostic ($R^2_{L \mid W;a} = 0$), the asymptotic relative efficiency between an efficient RAL estimator and the unadjusted estimator is equal to when the $L$ variable is not available to the efficient RAL estimator (i.e., when the efficient RAL estimator adjusts for $W$ alone). Similarly, when only $L$ is prognostic ($R^2_{W;a} = 0$), the asymptotic relative efficiency is equal to when the $W$ variable is not available to the efficient RAL estimator (i.e., when only adjusts for $L$ alone). In Figure \ref{fig:are-w} we illustrate the implication of Theorem \ref{cor:re} under those special cases. We plot $\are$ against $p_y$ when only $W$ is prognostic ($R^2_{L \mid W;a} = 0$), and in Figure \ref{fig:are-l} we plot $\are$ against $p_y / p_l$ when only $L$ is prognostic ($R^2_{W;a} = 0$). In each plot, we separately consider prognostic value being $0.1$ and $0.25$. We consider three implications of the curves in Figure \ref{fig:are}. First, the precision gain ($\are$) from adjusting for a prognostic baseline variable increases with a greater proportion of pipeline participants (i.e., smaller $p_y$). The reason is that each additional pipeline participant (who has observed baseline variables but missing outcome) contributes some information to the adjusted estimator (through the baseline variables) but no information to the unadjusted estimator. For a similar reason, the precision gain from adjusting for a prognostic short-term outcome increases with a greater proportion of participants with $W,L$ but not $Y$ observed  (i.e., smaller $p_y/p_l$).

Second, when every participant has their primary outcome observed (e.g., at the final analysis time of a group sequential design with no dropouts), adjusting for the prognostic baseline variables still improves estimation precision as long as $R^2_{W;a} > 0$, but adjusting for prognostic short-term outcome no longer does.

Third, for any given $(p_y,p_l)$, adjusting for a prognostic baseline variable alone always leads to larger precision gain than adjusting for an equally prognostic short-term outcome alone. 

Following the discussion in Section \ref{appen:are-sample-size-reduction}, we define the asymptotic equivalent reduction in sample size ($\redss$) of $\widehat{\Delta}_1$ compared to $\widehat{\Delta}_2$ as
\begin{equation}
\redss(\widehat{\Delta}_1, \widehat{\Delta}_2) = 1 - \are(\widehat{\Delta}_1, \widehat{\Delta}_2)^{-1}. \label{def:redss}
\end{equation}
Adjusting for prognostic $W$ alone with $R^2_{W;a}=q$ (for $0<q<1$) yields $\redss = q(1 - p_y/2)$, and adjusting for prognostic $L$ alone with $R^2_{L;a} = q$ yields $\redss = q(1 - p_y/p_l)$. The ratio of the two $\redss$s equals (assuming $q>0$)
\begin{equation}
r = \frac{\redss \text{ from adjusting only for prognostic } W \text{ with } R^2_{W;a}=q}{\redss \text{ from adjusting only for prognostic } L \text{ with } R^2_{L;a}=q} = \frac{1 - p_y/2}{1 - p_y/p_l}. \label{def:ratio-redss}
\end{equation}
$r>1$ means that the sample size reduction from adjusting for a prognostic $W$ is larger than that from adjusting for an equally prognostic $L$; $r<1$ means the opposite. Figure \ref{fig:contour} plots $r$ against $(p_l, p_y)$ under the constraint $0 < p_y \leq p_l \leq 1$. For all such $(p_l, p_y)$, $r > 1$. In addition, we have $r \geq 2$ if $p_y \geq 2/3$. This means that if at most $1/3$ of the enrollees are in the pipeline, the sample size reduction from adjusting for a prognostic short-term outcome alone is at most half of that from adjusting for an equally prognostic baseline variable alone.
Roughly speaking, small $p_y$ makes $r$ close to 1, meaning that adjusting for $W$ or adjusting for $L$ results in a similar sample size reduction when there are relatively few with $Y$ observed. This may occur at early stages of a group sequential design, if the delay of the primary outcome is long relative to the enrollment rate.

\begin{figure}[htbp]
\caption{Contour plot of $r(p_l, p_y)$, where $r$ is the ratio of the reductions in sample size when only $W$ is prognostic with $R^2_{W;a}=q$ and when only $L$ is prognostic with $R^2_{L;a}=q$, for any fixed $q \in (0,1)$. The estimand is the treatment specific mean $E(Y\mid A=a)$.}
\label{fig:contour}
\begin{center}
\includegraphics[width = 0.7\textwidth]{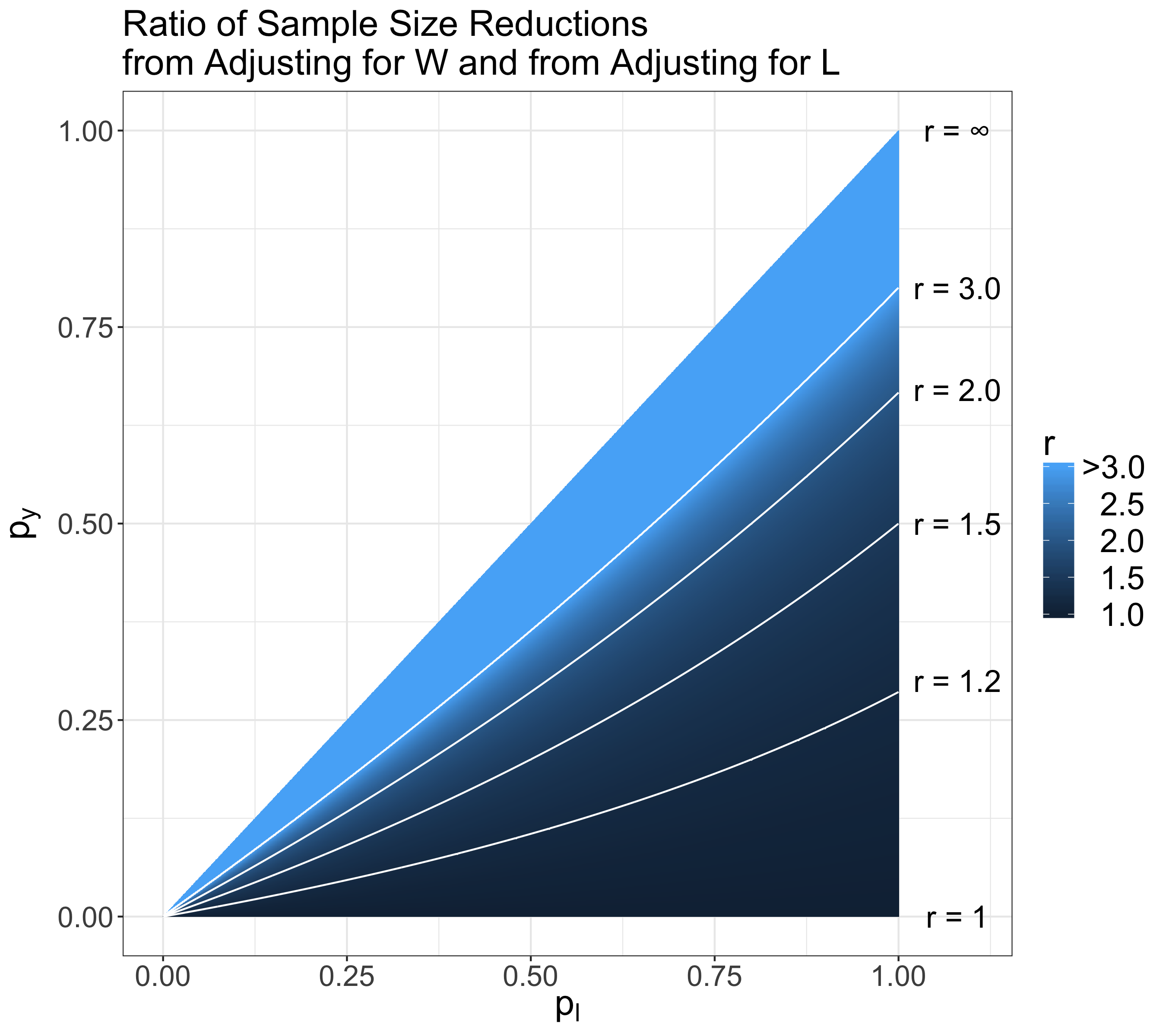}
\end{center}
\end{figure}

\section{Simulation Results for Relative Efficiency When Estimating the Treatment Specific Mean}
\label{appen:re-h1}

The simulation setup, including the data generating distributions (settings), are the same as Section \ref{sec:simulation}. 
Table \ref{tab:rsq-mistie-a0a1} gives the values of $R^2_{W;a}$ and $R^2_{L \mid W;a}$ for $a\in\{0,1\}$, which are defined in the paragraph after Theorem \ref{thm:avar}. Tables~\ref{tab:re-H0} and~\ref{tab:re-H1} list the asymptotic relative efficiency (ARE) approximated by substituting model-based estimates of 
the conditional expectations in the quantities in (\ref{eq:re}).

\begin{table}[htbp]

\caption{$R^2_{W;a}$ and $R^2_{L \mid W;a}$ for $a\in\{0,1\}$ for the data generating distributions from the settings in Section \ref{sec:simulation}, approximated based on a simulated dataset with 1,000,000 participants under each setting.}
\label{tab:rsq-mistie-a0a1}
\begin{center}
\begin{tabular}{llllllllllll}
\hline
                 & \multicolumn{2}{l}{$\progwl$} & & \multicolumn{2}{l}{$\progw$} & & \multicolumn{2}{l}{$\progl$} & & \multicolumn{2}{l}{$\prognon$} \\
                 & $\Delta=0$         & $\Delta=0.122$       &  & $\Delta=0$         & $\Delta=0.122$       & & $\Delta=0$         & $\Delta=0.122$        & & $\Delta=0$          & $\Delta=0.122$         \\
\hline
$R^2_{W;0}$        & 0.35          & 0.33        &  & 0.35          & 0.33        & & 0             & 0            & & 0              & 0             \\
$R^2_{L \mid W;0}$ & 0.08          & 0.06         & & 0             & 0          &  & 0.30          & 0.28         & & 0              & 0             \\
$R^2_{W;1}$        & 0.35          & 0.38        &  & 0.35          & 0.38        & & 0             & 0            & & 0              & 0             \\
$R^2_{L \mid W;1}$ & 0.08          & 0.08         & & 0             & 0          &  & 0.30          & 0.31         & & 0              & 0             \\
\hline
\end{tabular}
\end{center}
\end{table}

\begin{table}[htbp]
\caption{Comparison of the asymptotic relative efficiency (ARE) predicted by Theorem \ref{cor:re} and the relative efficiency computed from simulated trials under $\Delta=0$. The maximum sample size $\nmax$ is set to be 300 under $\progwl$ and $\progw$, and is set to 480 under $\progl$ and $\prognon$. Under each setting the same $\nmax$ is used for both the unadjusted estimator and the adjusted estimator.
The simulated RE is based on 50,000 simulated trials.} 
\label{tab:re-H0}
\begin{center}
\resizebox{\linewidth}{!}{
\begin{tabular}{lrrrrrrrrrr}
  \hline
Under $\Delta=0$ & & \multicolumn{4}{c}{ARE from theory}  & & \multicolumn{4}{c}{RE from Simulation} \\
  \cline{3-6} \cline{8-11}
& & $\progwl$ & $\progw$ & $\progl$ & $\prognon$ & & $\progwl$ & $\progw$ & $\progl$ & $\prognon$ \\ 
  \hline
  \\
   &  & \multicolumn{9}{c}{\textit{Estimand: } $E(Y\mid A=0)$} \\
 \multirow{4}{*}{\begin{tabular}{l} Interim \\ Analysis \end{tabular}} 
   & 1 & 1.44 & 1.36 & 1.13 & 1.00 && 1.34 & 1.29 & 1.09 & 0.97 \\ 
   & 2 & 1.35 & 1.31 & 1.07 & 1.00 && 1.32 & 1.28 & 1.06 & 0.99 \\ 
   & 3 & 1.32 & 1.29 & 1.05 & 1.00 && 1.29 & 1.27 & 1.04 & 0.99 \\ 
   & 4 & 1.29 & 1.26 & 1.04 & 1.00 && 1.27 & 1.25 & 1.03 & 0.99 \\ 
  \\
 \multirow{5}{*}{\begin{tabular}{l} Decision \\ Analysis \end{tabular}} 
   & 1 & 1.21 & 1.21 & 1.00 & 1.00 && 1.19 & 1.20 & 0.99 & 0.99 \\ 
   & 2 & 1.21 & 1.21 & 1.00 & 1.00 && 1.19 & 1.20 & 0.99 & 0.99 \\ 
   & 3 & 1.21 & 1.21 & 1.00 & 1.00 && 1.20 & 1.20 & 0.99 & 0.99 \\ 
   & 4 & 1.21 & 1.21 & 1.00 & 1.00 && 1.20 & 1.21 & 1.00 & 1.00 \\ 
   & 5 & 1.21 & 1.21 & 1.00 & 1.00 && 1.20 & 1.21 & 1.00 & 1.00 \\ 
  \\
  & & \multicolumn{9}{c}{\textit{Estimand: } $E(Y\mid A=1)$} \\
 \multirow{4}{*}{\begin{tabular}{l} Interim \\ Analysis \end{tabular}} 
   & 1 & 1.44 & 1.37 & 1.13 & 1.00 && 1.34 & 1.28 & 1.10 & 0.97 \\ 
   & 2 & 1.36 & 1.31 & 1.07 & 1.00 && 1.33 & 1.29 & 1.06 & 0.99 \\ 
   & 3 & 1.32 & 1.29 & 1.05 & 1.00 && 1.31 & 1.27 & 1.04 & 0.99 \\ 
   & 4 & 1.29 & 1.26 & 1.04 & 1.00 && 1.28 & 1.26 & 1.03 & 0.99 \\ 
  \\
 \multirow{5}{*}{\begin{tabular}{l} Decision \\ Analysis \end{tabular}} 
   & 1 & 1.21 & 1.21 & 1.00 & 1.00 && 1.20 & 1.20 & 0.99 & 0.99 \\ 
   & 2 & 1.21 & 1.21 & 1.00 & 1.00 && 1.20 & 1.21 & 0.99 & 0.99 \\ 
   & 3 & 1.21 & 1.21 & 1.00 & 1.00 && 1.20 & 1.21 & 0.99 & 0.99 \\ 
   & 4 & 1.21 & 1.21 & 1.00 & 1.00 && 1.21 & 1.21 & 1.00 & 1.00 \\ 
   & 5 & 1.21 & 1.21 & 1.00 & 1.00 && 1.21 & 1.21 & 1.00 & 1.00 \\ 
   \hline \\
\end{tabular}
}
\end{center}
\end{table}

\begin{table}[htbp]
\caption{Comparison of the asymptotic relative efficiency (ARE) predicted by Theorem \ref{cor:re} and the relative efficiency computed from simulated trials under $\Delta=0.122$. The maximum sample size $\nmax$ is set to be 300 under $\progwl$ and $\progw$, and is set to 480 under $\progl$ and $\prognon$. Under each setting the same $\nmax$ is used for both the unadjusted estimator and the adjusted estimator.
The simulated RE is based on 50,000 simulated trials.} 
\label{tab:re-H1}
\begin{center}
\resizebox{\linewidth}{!}{
\begin{tabular}{lrrrrrrrrrr}
  \hline
Under $\Delta=0.122$ & & \multicolumn{4}{c}{RE approximated by Theory}  & & \multicolumn{4}{c}{RE from Simulation} \\
  \cline{3-6} \cline{8-11}
& & $\progwl$ & $\progw$ & $\progl$ & $\prognon$ & & $\progwl$ & $\progw$ & $\progl$ & $\prognon$ \\ 
  \hline
  \\
   &  & \multicolumn{9}{c}{\textit{Estimand: } $E(Y\mid A=0)$} \\
 \multirow{4}{*}{\begin{tabular}{l} Interim \\ Analysis \end{tabular}} 
   & 1 & 1.59 & 1.54 & 1.14 & 1.00 && 1.49 & 1.43 & 1.10 & 0.97 \\ 
   & 2 & 1.48 & 1.45 & 1.08 & 1.00 && 1.46 & 1.41 & 1.07 & 0.99 \\ 
   & 3 & 1.43 & 1.41 & 1.06 & 1.00 && 1.44 & 1.40 & 1.05 & 0.99 \\ 
   & 4 & 1.39 & 1.38 & 1.04 & 1.00 && 1.40 & 1.38 & 1.03 & 0.99 \\ 
  \\
 \multirow{5}{*}{\begin{tabular}{l} Decision \\ Analysis \end{tabular}} 
   & 1 & 1.30 & 1.30 & 1.00 & 1.00 && 1.28 & 1.28 & 0.99 & 0.99 \\ 
   & 2 & 1.30 & 1.30 & 1.00 & 1.00 && 1.29 & 1.29 & 0.99 & 0.99 \\ 
   & 3 & 1.30 & 1.30 & 1.00 & 1.00 && 1.30 & 1.30 & 0.99 & 0.99 \\ 
   & 4 & 1.30 & 1.30 & 1.00 & 1.00 && 1.30 & 1.30 & 1.00 & 1.00 \\ 
   & 5 & 1.30 & 1.30 & 1.00 & 1.00 && 1.30 & 1.30 & 1.00 & 1.00 \\ 
  \\
&  & \multicolumn{9}{c}{\textit{Estimand: } $E(Y\mid A=1)$} \\
 \multirow{4}{*}{\begin{tabular}{l} Interim \\ Analysis \end{tabular}} 
   & 1 & 1.37 & 1.30 & 1.11 & 1.00 && 1.26 & 1.21 & 1.08 & 0.97 \\ 
   & 2 & 1.30 & 1.26 & 1.07 & 1.00 && 1.26 & 1.22 & 1.05 & 0.99 \\ 
   & 3 & 1.27 & 1.24 & 1.05 & 1.00 && 1.24 & 1.21 & 1.04 & 0.99 \\ 
   & 4 & 1.24 & 1.22 & 1.04 & 1.00 && 1.23 & 1.19 & 1.03 & 0.99 \\ 
  \\
 \multirow{5}{*}{\begin{tabular}{l} Decision \\ Analysis \end{tabular}} 
   & 1 & 1.18 & 1.18 & 1.00 & 1.00 && 1.16 & 1.16 & 0.99 & 0.99 \\ 
   & 2 & 1.18 & 1.18 & 1.00 & 1.00 && 1.16 & 1.16 & 0.99 & 0.99 \\ 
   & 3 & 1.18 & 1.18 & 1.00 & 1.00 && 1.17 & 1.16 & 0.99 & 0.99 \\ 
   & 4 & 1.18 & 1.18 & 1.00 & 1.00 && 1.17 & 1.16 & 1.00 & 1.00 \\ 
   & 5 & 1.18 & 1.18 & 1.00 & 1.00 && 1.17 & 1.16 & 1.00 & 1.00 \\ 
   \hline \\
\end{tabular}
}
\end{center}
\end{table}

\section{Proofs}
\label{appen:sec:proof}

We prove the results in the main paper as well as results in Section \ref{appen:theory-predictive-prognostic} and Section \ref{sec:theory-singlearm}. Auxiliary lemmas (used in these proofs) are proved Section \ref{appen:proof-aux}.

\subsection{Identification of the average treatment effect}

Under Assumptions \ref{assump:randomization}-\ref{assump:monotone-censoring}, the average treatment effect can be expressed as follows:
\begin{align}
E(Y \mid A=1) - E(Y \mid A=0) & =& E [ E \{ E (Y \mid L,W,A=1,C^Y=1) \mid W,C^L=1,A=1\} ] \nonumber \\
&& -  E [ E \{ E (Y \mid L,W,A=0,C^Y=1) \mid W,C^L=1,A=0\} ]  \label{eq:gcomputation}
\end{align}

The proof of (\ref{eq:gcomputation}) is given as follows.

\begin{proof}

It suffices to show that for each $a\in\{0,1\}$ the following holds:
\begin{equation}
E(Y \mid A=a) = E [ E \{ E (Y \mid L,W,A=a,C^Y=1) \mid W,A=a,C^L=1\} ]. \label{eq:gcomputation-proofuse-1}
\end{equation}
Using the law of iterated expectation twice, we deduce
\begin{equation}
E(Y \mid A=a) = E [ E \{ E (Y \mid L,W,A=a) \mid W,A=a\} ]. \label{eq:gcomputation-proofuse-3}
\end{equation}
Let $f(L,W) = E (Y \mid L,W,A=a)$. By Assumption \ref{assump:independent-censoring} we have
\begin{equation}
E(f(L,W) \mid A=a) = E(f(L,W) \mid A=a, C^L=1), \label{eq:gcomputation-proofuse-4}
\end{equation}
and
\begin{equation}
E(Y \mid L,W,A=a) = E(Y \mid L,W,A=a, C^L=1, C^Y=1). \label{eq:gcomputation-proofuse-5}
\end{equation}
Equations (\ref{eq:gcomputation-proofuse-3})-(\ref{eq:gcomputation-proofuse-5}) together yield
\begin{equation}
E(Y \mid A=a) = E [ E \{ E (Y \mid L,W,A=a,C^L=1,C^Y=1) \mid W,A=a,C^L=1\} ]. \label{eq:gcomputation-proofuse-6}
\end{equation}
Equation (\ref{eq:gcomputation-proofuse-6}) and Assumption \ref{assump:monotone-censoring} yield
(\ref{eq:gcomputation-proofuse-1}). This completes the proof.

\end{proof}

\subsection{Lemma on Variance Decomposition} \label{proof:lem:vardecomp} 

\begin{lemma}
\label{lem:vardecomp}

For study arm $a\in\{0,1\}$, we have the following decomposition of the variance of $Y$ in that arm:
\begin{equation}
\var_a(Y) = \var_a\{E_a(Y\mid W)\} + \var_a\{ E_a(Y\mid L,W) - E_a(Y \mid W) \} + \var_a\{Y - E_a(Y\mid L,W)\}. \label{eq:vardecomp}
\end{equation}
In addition, we have
\begin{equation}
R^2_{W;a} + R^2_{L\mid W;a} + R^2_{r;a} = 1, \quad \mbox{ for each } a \in \{0,1\}. \label{eq:rsqsum1}
\end{equation}
\end{lemma}

\begin{proof}

Adding and subtracting terms, we have
\begin{equation}
\var_a(Y) = \var_a \big[ \{Y - E_a(Y\mid L,W)\} + \{E_a(Y\mid L,W) - E_a(Y \mid W)\} + E_a(Y\mid W) \big].
\end{equation}
So for proving (\ref{eq:vardecomp}), it suffices to establish the following:
\begin{align}
\cov_a \big\{ Y - E_a(Y\mid L,W), E_a(Y\mid L,W) - E_a(Y \mid W) \big\} &= 0, \label{eq:vardecomp-proofuse-1} \\
\cov_a \big\{ Y - E_a(Y\mid L,W), E_a(Y\mid W) \big\} &= 0, \label{eq:vardecomp-proofuse-2} \\
\cov_a \big\{ E_a(Y\mid L,W) - E_a(Y \mid W), E_a(Y\mid W) \big\} &= 0, \label{eq:vardecomp-proofuse-3}
\end{align}
where $\cov_a$ denotes the conditional covariance given $A=a$.

First, we have
\begin{align}
& \cov_a \big\{ Y - E_a(Y\mid L,W), E_a(Y\mid L,W) - E_a(Y \mid W) \big\} \nonumber\\
= &  E_a \big\{ Y E_a(Y\mid L,W) - E_a(Y\mid L,W)^2
  - Y E_a(Y \mid W) + E_a(Y\mid L,W) E_a(Y \mid W) \big\}. \label{eq:vardecomp-proofuse-4}
\end{align}
By Lemma \ref{lem:cond-exp} with $X = Z = (L,W)$, we have
\begin{equation}
E_a \big\{ Y E_a(Y\mid L,W) \big\} = E_a \big\{ E_a(Y \mid L,W)^2 \big\}. \label{eq:vardecomp-proofuse-5}
\end{equation}
By Lemma \ref{lem:cond-exp} with $X = W$ and $Z = (L,W)$, we have
\begin{equation}
E_a \big\{ Y E_a(Y \mid W) \big\} = E_a \big\{ E_a(Y \mid L,W) E_a(Y\mid W) \big\}. \label{eq:vardecomp-proofuse-6}
\end{equation}
Equations (\ref{eq:vardecomp-proofuse-4}), (\ref{eq:vardecomp-proofuse-5}), and (\ref{eq:vardecomp-proofuse-6}) imply (\ref{eq:vardecomp-proofuse-1}).

Second, we have
\begin{equation}
\cov_a \big\{ Y - E_a(Y\mid L,W), E_a(Y\mid W) \big\} = E_a \big\{ Y E_a(Y\mid W) - E_a(Y\mid L,W) E_a(Y\mid W) \big\}. \label{eq:vardecomp-proofuse-7}
\end{equation}
Equations (\ref{eq:vardecomp-proofuse-7}) and (\ref{eq:vardecomp-proofuse-6}) imply (\ref{eq:vardecomp-proofuse-2}).

Third, since $E_a(Y\mid L,W) - E_a(Y \mid W)$ has expectation zero, we have
\begin{equation}
\cov_a \big\{ E_a(Y\mid L,W) - E_a(Y \mid W), E_a(Y\mid W) \big\}
= E_a \big\{ E_a(Y\mid L,W) E_a(Y\mid W) - E_a(Y \mid W)^2 \big\}. \label{eq:vardecomp-proofuse-8}
\end{equation}
In Lemma \ref{lem:cond-exp}, letting $X=Z=W$ and replacing $Y$ in the lemma by $E(Y \mid L,W)$ implies 
\begin{equation}
E_a \big\{ E_a(Y \mid L,W) E_a(Y\mid W) \big\} = E_a \big\{ E_a(Y \mid W)^2 \big\}. \label{eq:vardecomp-proofuse-9}
\end{equation}
Equations (\ref{eq:vardecomp-proofuse-8}) and (\ref{eq:vardecomp-proofuse-9}) imply (\ref{eq:vardecomp-proofuse-3}).

This proves (\ref{eq:vardecomp}). Equation (\ref{eq:rsqsum1}) follows immediately from (\ref{eq:vardecomp}) and the definition of $R^2_{W;a}$, $R^2_{L \mid W;a}$ and $R^2_{r;a}$. This completes the proof for Lemma \ref{lem:vardecomp}.

\end{proof}

\subsection{Proof of Theorem \ref{thm:avar}} \label{proof:thm:avar}

\begin{proof}

In the proof, we will use equation (\ref{eq:gcomputation-proofuse-6}) derived earlier; we rewrite it below:
\begin{equation}
E(Y \mid A=a) = E [ E \{ E (Y \mid L,W,A=a,C^L=1,C^Y=1) \mid W,A=a,C^L=1\} ]. \label{eq:thm:avar-proofuse-1}
\end{equation}
Treating the missingness indicators $C^L$ and $C^Y$ as binary treatments, the right-hand side of (\ref{eq:thm:avar-proofuse-1}) becomes the average of outcome $Y$ under time dependent treatment assignment: $A=a, C^L=1, C^Y=1$. Because there is no measurement made between $A$ and $C^L$, we can combine the two as a single treatment $\tilde{A}$, with $\tilde{A}=1$ if and only if $A=a$ and $C^L=1$. Equation (\ref{eq:thm:avar-proofuse-1}) becomes
\begin{equation}
E(Y \mid A=a) = E [ E \{ E (Y \mid L,W,\tilde{A}=1,C^Y=1) \mid W,\tilde{A}=1\} ]. \label{eq:thm:avar-proofuse-1.5}
\end{equation}
Using the fact that $L$ is binary-valued, by equations (24) and (28) in \citet{rosenblum2011EIF} or Theorem 1 in \citet{van2010EIF}, the efficient influence function for (\ref{eq:thm:avar-proofuse-1.5}) is
\begin{equation}
D(W,\tilde{A},L,C^Y,Y) = D_0(W) + D_1(W,\tilde{A},L) + D_2(W,\tilde{A},L,C^Y,Y), \label{eq:thm:avar-proofuse-d}
\end{equation}
where
\begin{equation}
D_0(W) = E(Y \mid W,\tilde{A}=1,C^Y=1) - E(Y \mid \tilde{A}=1), \label{eq:thm:avar-proofuse-d0}
\end{equation}
\begin{equation}
D_1(W,\tilde{A},L) = \frac{
\indic(\tilde{A}=1) \big\{ E(Y \mid L,W, \tilde{A}=1, C^Y=1) - E(Y \mid W,\tilde{A}=1,C^Y=1) \big\}
}{P(\tilde{A}=1\mid W)}, \label{eq:thm:avar-proofuse-d1}
\end{equation}
and
\begin{equation}
D_2(W,\tilde{A},L,C^Y,Y) = \frac{
\indic(C^Y=1) \indic(\tilde{A}=1) \big\{ Y - E(Y \mid L,W, \tilde{A}=1, C^Y=1) \big\}
}{P(C^Y=1 \mid L,W,\tilde{A}=1)P(\tilde{A}=1\mid W)}. \label{eq:thm:avar-proofuse-d2}
\end{equation}


By randomization and independent censoring assumptions, (\ref{eq:thm:avar-proofuse-d0}) simplifies to
\begin{equation}
D_0(W) = E(Y \mid W, A=a) - E(Y \mid A=a); \label{eq:thm:avar-proofuse-d0-simp}
\end{equation}
equation (\ref{eq:thm:avar-proofuse-d1}) simplifies to
\begin{equation}
D_1(W,\tilde{A},L) = \frac{
\indic(\tilde{A}=1) \big\{ E(Y \mid L,W, A=a) - E(Y \mid W,A=a) \big\}
}{P(\tilde{A}=1)}; \label{eq:thm:avar-proofuse-d1-simp}
\end{equation}
equation (\ref{eq:thm:avar-proofuse-d2}) simplifies to
\begin{equation}
D_2(W,\tilde{A},L,C^Y,Y) = \frac{
\indic(C^Y=1) \indic(\tilde{A}=1) \big\{ Y - E(Y \mid L,W,A=a) \big\}
}{P(C^Y=1 \mid C^L=1)P(\tilde{A}=1)}. \label{eq:thm:avar-proofuse-d2-simp}
\end{equation}

The following lemma states that $D_0$, $D_1$, and $D_2$ are pairwise uncorrelated.

\begin{lemma} \label{lem:thm-proofuse}

We have
\begin{align}
\cov\{ D_0(W), D_1(W,\tilde{A},L)\} & = 0, \label{eq:lem:thm-proofuse-toshow-d0d1}\\
\cov\{ D_0(W), D_2(W,\tilde{A},L,C^Y,Y)\} & = 0, \label{eq:lem:thm-proofuse-toshow-d0d2}\\
\cov\{ D_1(W,\tilde{A},L), D_2(W,\tilde{A},L,C^Y,Y)\} & = 0. \label{eq:lem:thm-proofuse-toshow-d1d2}
\end{align}

\end{lemma}

\bigskip

Lemma \ref{lem:thm-proofuse} implies
\begin{equation}
\var\{D(W,\tilde{A},L,C^Y,Y)\} = \var\{D_0(W)\} + \var\{D_1(W,\tilde{A},L)\} + \var\{D_2(W,\tilde{A},L,C^Y,Y)\}. \label{eq:thm:avar-proofuse-2}
\end{equation}
By (\ref{eq:thm:avar-proofuse-d0-simp}) we have
\begin{equation}
\var\{D_0(W)\} = \var\{E_a(Y \mid W)\} = \var_a\{E_a(Y \mid W)\}, \label{eq:thm:avar-proofuse-vard0}
\end{equation}
where the last equality follows from randomization assumption.
By (\ref{eq:thm:avar-proofuse-d1-simp}) we have $E\{D_1(W,\tilde{A},L)\}=0$, so it follows from randomization and independent censoring that
\begin{align}
\var\{D_1(W,\tilde{A},L)\} & = E \Big[ \frac{ \indic(\tilde{A}=1)^2 \big\{ E(Y \mid L,W, A=a) - E(Y \mid W,A=a) \big\}^2 }{P(\tilde{A}=1)^2} \Big] \nonumber \\
& = \frac{E\{\indic(\tilde{A}=1)\}}{P(\tilde{A}=1)^2} E\big[ \big\{ E(Y \mid L,W, A=a) - E(Y \mid W,A=a) \big\}^2 \mid A=a \big]. \label{eq:thm:avar-proofuse-vard1-1}
\end{align}
By independent censoring we have $P(\tilde{A}=1) = p_a p_l$. It then follows from (\ref{eq:thm:avar-proofuse-vard1-1}) and randomization assumption that
\begin{equation}
\var\{D_1(W,\tilde{A},L)\} = \frac{1}{p_a p_l} \var_a\big\{ E_a(Y \mid L,W) - E_a(Y \mid W) \big\}. \label{eq:thm:avar-proofuse-vard1}
\end{equation}
Similarly, (\ref{eq:thm:avar-proofuse-d2-simp}) together with randomization and monotone censoring imply
\begin{equation}
\var\{D_2(W,\tilde{A},L,C^Y,Y)\} = \frac{1}{p_a p_y} \var_a\big\{ Y - E_a(Y \mid L,W) \big\}. \label{eq:thm:avar-proofuse-vard2}
\end{equation}

Because the semiparametric lower bound on the asymptotic variance for an estimand equals the variance of the efficient influence function, by (\ref{eq:thm:avar-proofuse-2}), (\ref{eq:thm:avar-proofuse-vard0}), (\ref{eq:thm:avar-proofuse-vard1}), and (\ref{eq:thm:avar-proofuse-vard2}) we proved (\ref{eq:thm:var-bound}).

\end{proof}

\subsection{Proof of Theorem \ref{cor:re}} \label{proof:cor:re}

\begin{proof}

The unadjusted estimator $\widehat{\tau}$ for $E(Y \mid A=a)$ is defined as
\begin{equation}
\widehat{\tau} = \frac{\sum_{i=1}^n Y_i \indic(A_i = a, C^Y_i = 1)}{\sum_{i=1}^n \indic(A_i = a, C^Y_i = 1)}. \label{def:unadj-noncausal}
\end{equation}
Under Assumptions \ref{assump:randomization} and \ref{assump:independent-censoring}, $\widehat{\tau}$ is unbiased:
\begin{align}
E(\widehat{\tau}) = E \Big[ E\Big\{ \frac{\sum_{i=1}^n Y_i \indic(A_i = a, C^Y_i = 1)}{\sum_{i=1}^n \indic(A_i = a, C^Y_i = 1)} \Big| A_1,\ldots,A_n,C_{Y_1},\ldots,C_{Y_n}  \Big\} \Big] = E(Y \mid A = a).
\end{align}

In the following we calculate the asymptotic variance of $\widehat{\tau}$.
\begin{align}
\sqrt{n}\{\widehat{\tau} - E(Y \mid A=a)\} & = \frac{ \frac{1}{\sqrt{n}}\sum_{i=1}^n Y_i \indic(A_i = a, C^Y_i = 1)}{ \frac{1}{n}\sum_{i=1}^n \indic(A_i = a, C^Y_i = 1)} - \sqrt{n}E(Y \mid A=a) \nonumber \\
&= \frac{ \frac{1}{\sqrt{n}}\sum_{i=1}^n \{Y_i - E(Y \mid A=a)\} \indic(A_i = a, C^Y_i = 1)}{ \frac{1}{n}\sum_{i=1}^n \indic(A_i = a, C^Y_i = 1)}. \label{eq:unadj-asymp-1}
\end{align}
By Weak Law of Large Numbers and the independent censoring assumption,
\begin{equation}
\frac{1}{n}\sum_{i=1}^n \indic(A_i = a, C^Y_i = 1) \stackrel{P}{\to} p_a p_y, \label{eq:unadj-asymp-2}
\end{equation}
where $\stackrel{P}{\to}$ denotes convergence in probability.
By Central Limit Theorem,
\begin{equation}
\frac{1}{\sqrt{n}}\sum_{i=1}^n \{Y_i - E(Y \mid A=a)\} \indic(A_i = a, C^Y_i = 1) \stackrel{d}{\to} N(0, \sigma^2), \label{eq:unadj-asymp-3}
\end{equation}
where by randomization and independent censoring we have
\begin{align}
\sigma^2 & = \var \big[ \{Y - E(Y \mid A=a)\} \indic(A = a, C^Y = 1) \big] = E \big[ \{Y - E(Y \mid A=a)\}^2 \indic(A = a, C^Y = 1)^2 \big] \nonumber \\
& = p_a p_y \var(Y \mid A=a). \label{eq:unadj-asymp-4}
\end{align}
Combining (\ref{eq:unadj-asymp-1})-(\ref{eq:unadj-asymp-4}), it follows from Slutsky's theorem that
\begin{equation}
\sqrt{n}\{\widehat{\tau} - E(Y \mid A=a)\} \stackrel{d}{\to} N\big(0, (p_a p_y)^{-1} \var(Y \mid A=a)\big). \nonumber
\end{equation}
So the asymptotic variance of $\widehat{\tau}$ is $(p_a p_y)^{-1} \var(Y \mid A=a)$, which by randomization yields (recall that by definition $\var_a(Y) = \var(Y \mid A=a)$)
\begin{equation}
\avar(\unadj) = \frac{1}{p_a p_y} \var_a(Y). \label{eq:cor-proofuse-1}
\end{equation}
Equations (\ref{eq:cor-proofuse-1}) and (\ref{eq:rsqsum1}) imply
\begin{equation}
\avar(\unadj) = \frac{1}{p_a p_y} \var_a(Y) (R^2_{W;a} + R^2_{L\mid W;a} + R^2_{r;a}). \label{eq:cor-proofuse-2}
\end{equation}

On the other hand, Theorem \ref{thm:avar} and the definition of $R^2_W$ and $R^2_r$ imply
\begin{equation}
\avar(\eff) = \var_a(Y) ( R^2_{W;a} + \frac{1}{p_a p_l} R^2_{L\mid W;a} + \frac{1}{p_a p_y} R^2_{r;a} ). \label{eq:cor-proofuse-3}
\end{equation}
Equations (\ref{eq:cor-proofuse-2}), (\ref{eq:cor-proofuse-3}), and (\ref{eq:rsqsum1}) yield (\ref{eq:re}). The proof is thus finished.

\end{proof}

\subsection{Generalization and proof of Lemma \ref{thm:avar-ate}} \label{proof:thm:avar-ate}

We provide proof for a generalization of Lemma \ref{thm:avar-ate}, which allows for constant randomization probability other than $1/2$.

\begin{lemma} 
\label{thm:avar-ate-2}
Suppose Assumptions \ref{assump:randomization}, \ref{assump:independent-censoring}, and \ref{assump:monotone-censoring} hold. Define $P(A=a) = p_a$ for $a \in \{0,1\}$. The lower bound on the asymptotic variance of all RAL estimators of $E(Y\mid A=1) - E(Y\mid A=0)$ is 
\begin{align}
   & \var\{E_1(Y \mid W) - E_0(Y \mid W) \}
+  \sum_{a\in\{0,1\}} \frac{1}{p_a p_l} \var_a\{ E_a(Y \mid L,W) - E_a(Y \mid W) \} \nonumber \\
+ &  \sum_{a\in\{0,1\}} \frac{1}{p_a p_y} \var_a\{ Y - E_a(Y \mid L,W) \}. \label{eq:thm:var-bound-ate-2}
\end{align}

\end{lemma}

\begin{proof}

For notation simplicity, denote by $Q(\cdot)$ the conditional expectation $E(Y\mid \cdot)$. Using the derivation in (\ref{eq:thm:avar-proofuse-1})-(\ref{eq:thm:avar-proofuse-d2-simp}) twice for $A=1$ and $A=0$, we get the efficient influence function $D(W,A,C^L,L,C^Y,Y)$ for $E(Y \mid A=1) - E(Y \mid A=0)$:
\begin{equation}
D(W,A,C^L,L,C^Y,Y) = D_0(W) + D_1(W,A,C^L,L) + D_2(W,A,C^L,L,C^Y,Y),
\end{equation}
where
\begin{equation}
D_0(W) = \big\{ Q(W,A=1) - Q(A=1) \big\} - \big\{ Q(W,A=0) - Q(A=0) \big\}, \label{eq:thm:avar-ate-proofuse-d0}
\end{equation}
\begin{align}
D_1(W,A,C^L,L) = & \frac{A C^L}{p_1 p_l} \big\{ Q(W,L,A=1) - Q(W,A=1) \big\} \nonumber \\
& - \frac{(1-A)C^L}{p_0 p_l} \big\{ Q(W,L,A=0) - Q(W,A=0) \big\} , \label{eq:thm:avar-ate-proofuse-d1}
\end{align}
and
\begin{equation}
D_2(W,A,C^L,L,C^Y,Y) = \frac{A C^Y}{p_1 p_y} \big\{ Y - Q(W,L,A=1) \big\} - \frac{(1-A) C^Y}{p_0 p_y} \big\{ Y - Q(W,L,A=0) \big\}. \label{eq:thm:avar-ate-proofuse-d2}
\end{equation}

The following lemma states that $D_0$, $D_1$, and $D_2$ are pairwise uncorrelated.

\begin{lemma} \label{lem:thm-ate-proofuse}

We have
\begin{align}
\cov\{ D_0(W), D_1(W,A,C^L,L)\} & = 0, \label{eq:lem:thm-ate-proofuse-toshow-d0d1}\\
\cov\{ D_0(W), D_2(W,A,C^L,L,C^Y,Y)\} & = 0, \label{eq:lem:thm-ate-proofuse-toshow-d0d2}\\
\cov\{ D_1(W,A,C^L,L), D_2(W,A,C^L,L,C^Y,Y)\} & = 0. \label{eq:lem:thm-ate-proofuse-toshow-d1d2}
\end{align}

\end{lemma}

\bigskip

Lemma \ref{lem:thm-ate-proofuse} implies
\begin{align}
\var\{D(W,A,C^L,L,C^Y,Y)\} = & \var\{D_0(W)\} + \var\{D_1(W,A,C^L,L)\} \nonumber \\
& + \var\{D_2(W,A,C^L,L,C^Y,Y)\}. \label{eq:thm:avar-ate-proofuse-2}
\end{align}
By \ref{eq:thm:avar-ate-proofuse-d0} we have
\begin{equation}
\var\{D_0(W)\} = \var\{Q(W,A=1) - Q(W,A=0)\}. \label{eq:thm:avar-ate-proofuse-vard0}
\end{equation}
By (\ref{eq:thm:avar-ate-proofuse-d1}) we have
\begin{align}
\var\{D_1(W,A,C^L,L)\} = & E \Big[ \frac{A C^L}{p_1^2 p_l^2} \big\{ Q(W,L,A=1) - Q(W,A=1) \big\}^2 \Big] \nonumber \\
& + E \Big[ \frac{(1-A) C^L}{p_0^2 p_l^2} \big\{ Q(W,L,A=0) - Q(W,A=0) \big\}^2 \Big] \label{eq:thm:avar-ate-proofuse-vard1-1} \\
& = \sum_{a\in\{0,1\}} \frac{1}{p_a p_l} E \Big[ \big\{ Q(W,L,A=a) - Q(W,A=a) \big\}^2 \mid A=a \Big] \label{eq:thm:avar-ate-proofuse-vard1-2} \\
& = \sum_{a\in\{0,1\}} \frac{1}{p_a p_l} \var \big\{ Q(W,L,A=a) - Q(W,A=a) \mid A=a \big\}. \label{eq:thm:avar-ate-proofuse-vard1}
\end{align}
The step from (\ref{eq:thm:avar-ate-proofuse-vard1-1}) to (\ref{eq:thm:avar-ate-proofuse-vard1-2}) utilizes the independent censoring and randomization assumptions.
Similarly, (\ref{eq:thm:avar-ate-proofuse-d2}) together with randomization and independent censoring imply
\begin{equation}
\var\{D_2(W,A,C^L,L,C^Y,Y)\} = \sum_{a\in\{0,1\}} \frac{1}{p_a p_l} \var \big\{ Y - Q(W,L,A=a) \mid A=a \big\}. \label{eq:thm:avar-ate-proofuse-vard2}
\end{equation}

Because the semiparametric lower bound on the asymptotic variance for an estimand equals the variance of the efficient influence function, by (\ref{eq:thm:avar-ate-proofuse-2}), (\ref{eq:thm:avar-ate-proofuse-vard0}), (\ref{eq:thm:avar-ate-proofuse-vard1}), and (\ref{eq:thm:avar-ate-proofuse-vard2}) we proved Lemma \ref{thm:avar-ate}.

\end{proof}

\subsection{Proof of Theorem \ref{cor:re-ate}} \label{proof:cor:re-ate} 

\begin{proof}

The unadjusted estimator for the average treatment effect is
\begin{equation}
\widehat{\tau} = \frac{\sum_{i=1}^n Y_i \indic(A_i = 1, C^Y_i = 1)}{\sum_{i=1}^n \indic(A_i = 1, C^Y_i = 1)} - \frac{\sum_{i=1}^n Y_i \indic(A_i = 0, C^Y_i = 1)}{\sum_{i=1}^n \indic(A_i = 0, C^Y_i = 1)}. \nonumber
\end{equation}
Similar to the derivation from (\ref{def:unadj-noncausal}) to (\ref{eq:cor-proofuse-1}), when estimating the average treatment effect with $P(A=1) = P(A=0) = 1/2$, we have
\begin{equation}
\avar(\unadj) = \frac{2}{p_y} \sum_{a \in \{0,1\}} \var_a(Y).
\end{equation}
The result in Theorem \ref{cor:re-ate} then follows immediately from Lemma \ref{thm:avar-ate}.

\end{proof}

\subsection{Proof of Corollary \ref{cor:prognostic-predictive}}
\label{proof:cor:prognostic-predictive}

\begin{proof}

\begin{enumerate}

\item[(i)] Because
\begin{align}
\var\{E(Y \mid W)\} &= \var\{0.5 E(Y \mid W, A=1) + 0.5 E(Y \mid W, A=1)\}, \nonumber
\end{align}
$\var\{E(Y \mid W)\} = 0$ implies
\begin{align}
\var\{E(Y \mid W, A=1)\} + \var\{E(Y \mid W, A=0)\} = -2\cov\{E(Y \mid W, A=1),E(Y \mid W, A=0). \label{eq:proofuse:cor:prognostic-predictive-1}
\end{align}
By (\ref{eq:proofuse:cor:prognostic-predictive-1}) and the definition of $\gamma$ and $R^2_W$, we have $\gamma = 2 R^2_W$.
This combined with (\ref{eq:re-ate}) with $R^2_{L \mid W} = 0$ imply (\ref{eq:rsq-onlyW-ate-predictive}).

\item[(ii)] By the definition of $\gamma$ and $\var\{E(Y \mid W, A=1) - E(Y \mid W, A=0)\} = 0$, (\ref{eq:rsq-onlyW-ate-prognostic}) follows immediately from (\ref{eq:re-ate}) with $R^2_{L \mid W} = 0$.

\end{enumerate}

\end{proof}

\subsection{Proof of statements in Section \ref{impactcomparison}}
\label{appen:proof-section3.4}

For ease of reading we restate the claims to be proven in Section \ref{impactcomparison}: We compare the $\are$ between two cases: $R^2_W = q >0, R^2_{L \mid W} = 0$ (only baseline variable prognostic) and 
$R^2_W = 0, R^2_{L \mid W} = q >0$ (only short-term outcome prognostic).  Regardless of the value of $q>0$, the $\are$ in the 
former case is larger or equal to that in the latter case. 
Equality occurs if and only if  
$p_l=1$ and  $W$ is  uncorrelated with $Y$ (marginally) in the former case. 
The latter condition is equivalent to 
the treatment effect heterogeneity being the maximum possible $\gamma=2R^2_W=2q$.

\begin{proof}
We first show that $\gamma \leq 2R^2_W$ with equality holds only when $E_1(Y\mid W) = - E_0(Y \mid W)$ almost surely. By Cauchy-Schwarz inequality we have
\begin{align*}
    & \var\{ E_1(Y \mid W) - E_0(Y \mid W) \} \\
    & = \var\{ E_1(Y \mid W)\} + \var\{ E_0(Y \mid W) \} - 2 \cov\{E_1(Y \mid W), E_0(Y \mid W) \} \\
    & \leq \var\{ E_1(Y \mid W)\} + \var\{ E_0(Y \mid W) \} + 2 \left[\var\{ E_1(Y \mid W)\}\var\{ E_0(Y \mid W) \}\right]^{1/2} \\
    & \leq 2 \var\{ E_1(Y \mid W)\} + 2 \var\{ E_0(Y \mid W) \},
\end{align*}
where both inequalities becomes equality if and only if $E_1(Y \mid W) = - E_0(Y \mid W)$ almost surely. This proves the claim.

Theorem \ref{cor:re-ate} states that the $\are$ equals (where we indicate dependence on the arguments by explicitly writing them out)
\begin{equation}
\are(R^2_W, R^2_{L\mid W}, \gamma, p_y, p_l) = \frac{1}{ 1 + (p_y/2) \gamma - R^2_W -  (1 - p_y / p_l) R^2_{L\mid W} }. \nonumber
\end{equation}
When $R^2_W = q >0, R^2_{L \mid W} = 0$,
\begin{align*}
    \are(q, 0, \gamma, p_y, p_l) = \frac{1}{ 1 + (p_y/2) \gamma - q } \geq \frac{1}{1 - (1 - p_y)q},
\end{align*}
with equality holds if and only if $\gamma = 2q$.
When $R^2_W = 0, R^2_{L \mid W} = q >0$ (which implies $\gamma = 0$ because $\gamma \leq 2R^2_W$),
\begin{equation}
\are(0, R^2_{L\mid W}, 0, p_y, p_l) = \frac{1}{ 1 -  (1 - p_y / p_l) q }. \nonumber
\end{equation}
Because $p_l \leq 1$, $\are(q, 0, \gamma, p_y, p_l) \geq \are(0, R^2_{L\mid W}, 0, p_y, p_l)$ with equality holds if and only if $p_l = 1$ and $\gamma = 2q$, i.e., $E_1(Y\mid W) = - E_0(Y \mid W)$ almost surely.

Lastly, we show that $E_1(Y\mid W) = - E_0(Y \mid W)$ almost surely implies that $W$ is uncorrelated with $Y$ (marginally):
\begin{align*}
    & \cov(Y, W) = E\{Y (W - E(W))\} = E\{E(Y \mid W) (W - E(W))\} \\
    & = E\{E(Y \mid W) (W - E(W))\} = E\left[\{0.5 E_1(Y\mid W) + 0.5 E_0(Y\mid W)\} (W - E(W))\right] = 0.
\end{align*}
This completes the proof.

\end{proof}

\section{Proof of Auxiliary Lemmas}
\label{appen:proof-aux}

\subsection{Additional Supporting Lemmas}

\begin{lemma} \label{lem:cond-exp}

Consider three random variables $X$, $Y$, and $Z$. Denote by $\sigma(Z)$ the $\sigma$-field generated by $Z$. If $X \in \sigma(Z)$, then
\begin{equation}
E\big\{Y E(Y \mid X)\big\} = E\big\{ E(Y\mid Z) E(Y\mid X) \big\}. \label{eq:cond-exp}
\end{equation}
\end{lemma}

\bigskip

\begin{proof}

By the law of iterated expectation, we have
\begin{equation}
E\big\{Y E(Y \mid X)\big\} = E\big[ E \big\{Y E(Y \mid X) \mid Z\big\} \big]. \label{eq:cond-exp-proofuse-1}
\end{equation}
Because $X \in \sigma(Z)$, we have $E(Y \mid X) \in \sigma(X) \subset \sigma(Z)$. This implies
\begin{equation}
E \big\{Y E(Y \mid X) \mid Z\big\} = E(Y \mid X) E (Y \mid Z). \label{eq:cond-exp-proofuse-2}
\end{equation}
Equations (\ref{eq:cond-exp-proofuse-1}) and (\ref{eq:cond-exp-proofuse-2}) imply (\ref{eq:cond-exp}). This completes the proof.

\end{proof}

\begin{lemma} \label{lem:cond-exp-ate}

Consider three random variables $W$, $L$, and $Y$. For any measurable functions $f(W)$ and $g(W,L)$, we have
\begin{align}
E\big[ f(W) \big\{ E(Y \mid W,L) - E(Y \mid W) \big\} \big] & = 0, \label{eq:cond-exp-ate-1} \\
E\big[ g(W,L) \big\{ Y - E(Y \mid W,L) \big\} \big] & = 0. \label{eq:cond-exp-ate-2}
\end{align}

\end{lemma}

\bigskip

\begin{proof}

By the law of iterated expectation, we have
\begin{align}
E \big\{ g(W) E(Y \mid W,L) \big\} &= E \big[ E \big\{ g(W) E(Y \mid W,L) \mid W \big\} \big ] \nonumber \\
&= E \big[ g(W) E \big\{ E(Y \mid W,L) \mid W \big\} \big ] \nonumber \\
&= E \big\{ g(W) E(Y \mid W) \big\}, \nonumber
\end{align}
which proves (\ref{eq:cond-exp-ate-1}).

Similarly, we have
\begin{equation}
E \big\{ g(W,L) Y \big\} = E \big[ E \big\{ g(W,L) Y \mid W,L \big\} \big] = E \big\{ g(W,L) E(Y \mid W,L) \big\}, \nonumber
\end{equation}
which proves (\ref{eq:cond-exp-ate-2}).

\end{proof}

\subsection{Proof of Lemma \ref{lem:thm-proofuse}} \label{proof:lem:thm-proofuse}

\begin{proof}

By (\ref{eq:thm:avar-proofuse-d0-simp}) and (\ref{eq:thm:avar-proofuse-d1-simp}) we have
\begin{equation}
\cov(D_0, D_1) \propto E \big[ \indic(\tilde{A}=1)
\{ E_a(Y \mid W) - E_a(Y) \} \{ E_a(Y \mid L,W) - E_a(Y \mid W) \} \big]. \label{eq:lem:thm-proofuse-d0d1}
\end{equation}
Combining (\ref{eq:vardecomp-proofuse-3}), (\ref{eq:lem:thm-proofuse-d0d1}), and Assumptions \ref{assump:randomization} and \ref{assump:independent-censoring}, we derive (\ref{eq:lem:thm-proofuse-toshow-d0d1}).

By (\ref{eq:thm:avar-proofuse-d0-simp}) and (\ref{eq:thm:avar-proofuse-d2-simp}) we have
\begin{equation}
\cov(D_0, D_2) \propto E \big[ \indic(\tilde{A}=1,C^Y=1)
\{ E_a(Y \mid W) - E_a(Y) \} \{ Y - E_a(Y \mid L,W) \} \big]. \label{eq:lem:thm-proofuse-d0d2}
\end{equation}
Combining (\ref{eq:vardecomp-proofuse-2}), (\ref{eq:lem:thm-proofuse-d0d2}), and Assumptions \ref{assump:randomization} and \ref{assump:independent-censoring}, we derive (\ref{eq:lem:thm-proofuse-toshow-d0d2}).

By (\ref{eq:thm:avar-proofuse-d1-simp}) and (\ref{eq:thm:avar-proofuse-d2-simp}) we have
\begin{equation}
\cov(D_1, D_2) \propto E \big[ \indic(\tilde{A}=1,C^Y=1)
\{ E_a(Y \mid L,W) - E_a(Y \mid W) \} \{ Y - E_a(Y \mid L,W) \} \big]. \label{eq:lem:thm-proofuse-d1d2}
\end{equation}
Combining (\ref{eq:vardecomp-proofuse-1}), (\ref{eq:lem:thm-proofuse-d1d2}), and Assumptions \ref{assump:randomization} and \ref{assump:independent-censoring}, we derive (\ref{eq:lem:thm-proofuse-toshow-d1d2}).

This completes the proof.

\end{proof}

\subsection{Proof of Lemma \ref{lem:thm-ate-proofuse}} \label{proof:lem:thm-ate-proofuse}

For notation simplicity, we use $E_1(\cdot)$ and $E_0(\cdot)$ to denote $E(\cdot \mid A=1)$ and $E(\cdot \mid A=0)$, respectively.
By (\ref{eq:thm:avar-ate-proofuse-d0}) and (\ref{eq:thm:avar-ate-proofuse-d1}) we have
\begin{align}
\cov(D_0, D_1) \propto & E \big[ A \big\{ Q(W,A=1) - Q(W,A=0) \big\} \big\{ Q(W,L,A=1) - Q(W,A=1) \big\} \big] \nonumber \\
& - E \big[ (1-A) \big\{ Q(W,A=1) - Q(W,A=0) \big\} \big\{ Q(W,L,A=0) - Q(W,A=0) \big\} \big] \nonumber \\
&= \sum_{a\in\{0,1\}} E_a \big[ \big\{ E_1(Y \mid W) - E_0(Y \mid W) \big\} \big\{ E_a(Y \mid W,L) - E_a(Y \mid W) \big\} \times P(A=a). \label{eq:lem:thm-ate-proofuse-d0d1}
\end{align}
Both terms in (\ref{eq:lem:thm-ate-proofuse-d0d1}) equals 0 by (\ref{eq:cond-exp-ate-1}) in Lemma \ref{lem:cond-exp-ate} with $f(W) = E_1(Y \mid W) - E_0(Y \mid W)$. This yields (\ref{eq:lem:thm-ate-proofuse-toshow-d0d1}).

By (\ref{eq:thm:avar-ate-proofuse-d0}) and (\ref{eq:thm:avar-ate-proofuse-d2}) we have
\begin{align}
\cov(D_0, D_2) \propto & E \big[ A \big\{ Q(W,A=1) - Q(W,A=0) \big\} \big\{ Y - Q(W,L,A=1) \big\} \big] \nonumber \\
& - E \big[ (1-A) \big\{ Q(W,A=1) - Q(W,A=0) \big\} \big\{ Y - Q(W,L,A=0) \big\} \big] \nonumber \\
&= \sum_{a\in\{0,1\}} E_a \big[ \big\{ E_1(Y \mid W) - E_0(Y \mid W) \big\} \big\{ Y - E_a(Y \mid W,L) \big\} \times P(A=a). \label{eq:lem:thm-ate-proofuse-d0d2}
\end{align}
Both terms in (\ref{eq:lem:thm-ate-proofuse-d0d2}) equals 0 by (\ref{eq:cond-exp-ate-2}) in Lemma \ref{lem:cond-exp-ate} with $g(W,L) = E_1(Y \mid W) - E_0(Y \mid W)$. This yields (\ref{eq:lem:thm-ate-proofuse-toshow-d0d2}).

By (\ref{eq:thm:avar-ate-proofuse-d1}) and (\ref{eq:thm:avar-ate-proofuse-d2}) we have
\begin{align}
\cov(D_1, D_2) \propto & E \big[ A \big\{ Q(W,L,A=1) - Q(W,A=1) \big\} \big\{ Y - Q(W,L,A=1) \big\} \big] \nonumber \\
& - E \big[ (1-A) \big\{ Q(W,L,A=0) - Q(W,A=0) \big\} \big\{ Y - Q(W,L,A=0) \big\} \big] \nonumber \\
&= \sum_{a\in\{0,1\}} E_a \big[ \big\{ E_a(Y \mid W,L) - E_a(Y \mid W) \big\} \big\{ Y - E_a(Y \mid W,L) \big\} \times P(A=a). \label{eq:lem:thm-ate-proofuse-d1d2}
\end{align}
Both terms in (\ref{eq:lem:thm-ate-proofuse-d1d2}) equals 0 by (\ref{eq:cond-exp-ate-2}) in Lemma \ref{lem:cond-exp-ate} with $g(W,L) = E_a(Y \mid W,L) - E_a(Y \mid W)$. This yields (\ref{eq:lem:thm-ate-proofuse-toshow-d1d2}).

This completes the proof.

\end{document}